\def\UseBibLatex{1}

\ifx\GenSubmitVer\undefined%
\documentclass[12pt]{article}%
\newcommand{\SubmitVer}[1]{}
\newcommand{\RegVer}[1]{#1}
\else%
\documentclass[smallextended]{svjour3}
\newcommand{\SubmitVer}[1]{#1}
\newcommand{\RegVer}[1]{}
\fi%

\IfFileExists{sariel_computer.sty}{\def\sarielComp{1}}{}
\ifx\sarielComp\undefined%
\newcommand{\SarielComp}[1]{}
\newcommand{\NotSarielComp}[1]{#1}%
\else
\newcommand{\SarielComp}[1]{#1}%
\newcommand{\NotSarielComp}[1]{}%
\fi
\newcommand{\IfPrinterVer}[2]{#2}%

\usepackage{amsmath}%
\usepackage{amssymb}%
\usepackage{xcolor}%
\usepackage{wasysym}%
\usepackage{mathrsfs}
\usepackage{caption}

\DeclareFontFamily{U}{BOONDOX-calo}{\skewchar\font=45 }
\DeclareFontShape{U}{BOONDOX-calo}{m}{n}{
  <-> s*[1.05] BOONDOX-r-calo}{}
\DeclareFontShape{U}{BOONDOX-calo}{b}{n}{
  <-> s*[1.05] BOONDOX-b-calo}{}
\DeclareMathAlphabet{\mathcalb}{U}{BOONDOX-calo}{m}{n}
\SetMathAlphabet{\mathcalb}{bold}{U}{BOONDOX-calo}{b}{n}
\DeclareMathAlphabet{\mathbcalb}{U}{BOONDOX-calo}{b}{n}

\SarielComp{\usepackage{sariel_colors}}%

\RegVer{\usepackage[cm]{fullpage}}

\RegVer{
  \usepackage[amsmath,thmmarks]{ntheorem}%
  \theoremseparator{.}%
}

\SubmitVer{%
  \usepackage{amsthm}%
}%

\RegVer{
  \usepackage{titlesec}%
  \titlelabel{\thetitle. }%
}
\usepackage{xcolor}%
\usepackage{mleftright}%
\usepackage{xspace}%
\usepackage{graphicx}%
\usepackage{stmaryrd}%
\RegVer{\usepackage{euscript}}
\usepackage{multirow}

\IfPrinterVer{%
   \usepackage{hyperref}%
}{%
   \usepackage{hyperref}%
   \hypersetup{%
      breaklinks,%
      colorlinks=true,%
      urlcolor=[rgb]{0.25,0.0,0.0},%
      linkcolor=[rgb]{0.5,0.0,0.0},%
      citecolor=[rgb]{0,0.2,0.445},%
      filecolor=[rgb]{0,0,0.4},
      anchorcolor=[rgb]={0.0,0.1,0.2}%
   }
}

\numberwithin{figure}{section}%
\numberwithin{table}{section}%
\numberwithin{equation}{section}%

\SubmitVer{%
  \spnewtheorem*{remarks}{Remarks}{\it}{\rm}
}%

\RegVer{%
\theoremseparator{.}%

\theoremstyle{plain}%
\newtheorem{theorem}{Theorem}[section]

\newtheorem{lemma}[theorem]{Lemma}

\theoremstyle{plain}%
\theoremheaderfont{\sf} \theorembodyfont{\upshape}%
\newtheorem*{remark:unnumbered}[theorem]{Remark}%
\newtheorem*{remarks}[theorem]{Remarks}%
\newtheorem{remark}[theorem]{Remark}%
\newtheorem{definition}[theorem]{Definition}

\newcommand{\myqedsymbol}{\rule{2mm}{2mm}}

\theoremheaderfont{\em}%
\theorembodyfont{\upshape}%
\theoremstyle{nonumberplain}%
\theoremseparator{}%
\theoremsymbol{\myqedsymbol}%
\newtheorem{proof}{Proof:}%
}

\newcommand{\atgen}{\symbol{'100}}
\newcommand{\SarielThanks}[1]{\thanks{Department of Computer Science;
      University of Illinois; 201 N. Goodwin Avenue; Urbana, IL,
      61801, USA; {\tt sariel\atgen{}illinois.edu}; {\tt
         \url{http://sarielhp.org/}.} #1}}

\newcommand{\HLink}[2]{\hyperref[#2]{#1~\ref*{#2}}}
\newcommand{\HLinkSuffix}[3]{\hyperref[#2]{#1\ref*{#2}{#3}}}

\newcommand{\figlab}[1]{\label{fig:#1}}
\newcommand{\figref}[1]{\HLink{Figure}{fig:#1}}

\newcommand{\thmlab}[1]{{\label{theo:#1}}}
\newcommand{\thmref}[1]{\HLink{Theorem}{theo:#1}}

\newcommand{\apndlab}[1]{\label{apnd:#1}}
\newcommand{\apndref}[1]{\HLink{Appendix}{apnd:#1}}

\newcommand{\lemlab}[1]{\label{lemma:#1}}
\newcommand{\lemref}[1]{\HLink{Lemma}{lemma:#1}}%

\newcommand{\seclab}[1]{\label{sec:#1}}
\newcommand{\secref}[1]{\HLink{Section}{sec:#1}}

\providecommand{\deflab}[1]{\label{def:#1}}
\newcommand{\defref}[1]{\HLink{Definition}{def:#1}}
\newcommand{\defrefY}[2]{\hyperref[def:#2]{#1}}

\newcommand{\MitchellThanks}[1]{%
   \thanks{%
      Department of Computer Science;
      University of Illinois; 201 N. Goodwin Avenue; Urbana, IL,
      61801, USA; {\tt mfjones2\atgen{}illinois.edu}; {\tt
         \url{http://mfjones2.web.engr.illinois.edu/}.} #1}}

\newcommand{\RahulThanks}[1]{%
   \thanks{%
      Department of Computer Science and Automation;
      Indian Institute of Science; 
      Mathikere, Bengaluru, Karnataka 560012, India; 
      {\tt saladi\atgen{}iisc.ac.in}; {\tt
         \url{https://www.csa.iisc.ac.in/\~saladi/}.} #1}}

\providecommand{\eqlab}[1]{}%
\renewcommand{\eqlab}[1]{\label{equation:#1}}
\newcommand{\Eqref}[1]{\HLinkSuffix{Eq.~(}{equation:#1}{)}}

\newcommand{\remove}[1]{}%
\newcommand{\Set}[2]{\left\{ #1 \;\middle\vert\; #2 \right\}}
\newcommand{\pth}[2][\!]{\mleft({#2}\mright)}%
\newcommand{\pbrcx}[1]{\left[ {#1} \right]}%
\newcommand{\Prob}[1]{\mathop{\mathbf{Pr}}\!\pbrcx{#1}}
\newcommand{\Ex}[2][\!]{\mathop{\mathbf{E}}#1\pbrcx{#2}}

\newcommand{\ceil}[1]{\left\lceil {#1} \right\rceil}
\newcommand{\floor}[1]{\left\lfloor {#1} \right\rfloor}

\newcommand{\brc}[1]{\left\{ {#1} \right\}}
\newcommand{\cardin}[1]{\left| {#1} \right|}%

\newcommand{\tldO}{\widetilde{O}}

\renewcommand{\th}{th\xspace}
\usepackage[inline]{enumitem}

\newlist{compactenumA}{enumerate}{5}%
\setlist[compactenumA]{topsep=0pt,itemsep=-1ex,partopsep=1ex,parsep=1ex,%
   label=(\Alph*)}%

\newlist{compactenuma}{enumerate}{5}%
\setlist[compactenuma]{topsep=0pt,itemsep=-1ex,partopsep=1ex,parsep=1ex,%
   label=(\alph*)}%

\newlist{compactenumi}{enumerate}{5}%
\setlist[compactenumi]{topsep=0pt,itemsep=-1ex,partopsep=1ex,parsep=1ex,%
   label=(\roman*)}%

\newlist{compactenumi*}{enumerate*}{5}%
\setlist[compactenumi*]{topsep=0pt,itemsep=-1ex,partopsep=1ex,parsep=1ex,%
   label=(\roman*)}%

\definecolor{blue25emph}{rgb}{0, 0, 11}
\providecommand{\emphic}[2]{%
   \textcolor{blue25emph}{%
      \textbf{\emph{#1}}}%
   \index{#2}}

\providecommand{\emphi}[1]{\emphic{#1}{#1}}

\newcommand{\Caratheodory}{Carath\'eodory\xspace}
\newcommand{\Turan}{Tur\'an\xspace}

\newcommand{\etal}{et~al.\xspace~}
\providecommand{\Matousek}{Matou{\v s}ek\xspace}

\providecommand{\Mh}[1]{#1}%

\newcommand{\Interval}{\Mh{I}}%
\newcommand{\IntervalA}{\Mh{J}}%

\renewcommand{\Re}{\mathbb{R}}%

\newcommand{\eps}{\varepsilon}%

\newcommand{\inapx}{\Mh{B}}

\newcommand{\PS}{\Mh{P}}%
\newcommand{\USet}{\Mh{U}}%

\newcommand{\USetin}[1]{\Mh{U^{#1}_{\mathrm{in}}}}
\newcommand{\USetout}[1]{\Mh{U^{#1}_{\mathrm{out}}}}

\newcommand{\Sample}{\Mh{R}}%
\newcommand{\SSet}{\Mh{S}}%
\newcommand{\TSet}{T}

\newcommand{\PSA}{\Mh{T}}%
\newcommand{\PSB}{\Mh{U}}%
\renewcommand{\th}{th\xspace}

\newcommand{\body}{\Mh{C}}%

\newcommand{\indexX}[1]{\Mh{\varhexagon}^{}_{#1}}%

\newcommand{\pdisk}{\Mh{\mathcalb{d}}}%

\newcommand{\pnt}{\Mh{p}}%
\newcommand{\pntB}{\Mh{r}}
\newcommand{\pntC}{\Mh{u}}
\newcommand{\pntD}{\Mh{v}}
\newcommand{\pntq}{\Mh{q}}

\newcommand{\pa}{p}
\newcommand{\pb}{q}
\newcommand{\ps}{s}

\newcommand{\query}{\Mh{q}}%

\newcommand{\pntx}{\Mh{x}}%

\newcommand{\cpnt}{\Mh{\mathcalb{c}}}%

\newcommand{\indep}{\Mh{\alpha}}%
\newcommand{\clique}{\Mh{\omega}}%

\newcommand{\RSet}{\Mh{R}}%
\newcommand{\IY}[2]{\Mh{I}_{#1}(\pnt)} \newcommand{\IX}[1]{\Mh{I}(#1)}

\newcommand{\depthX}[1]{\mathrm{depth}\pth{#1}}

\newcommand{\SegSet}{\Mh{S}}%

\DefineNamedColor{named}{AlgorithmColor}{cmyk}{0.07,0.90,0,0.34}

\newcommand{\AlgorithmI}[1]{{%
      \textcolor[named]{AlgorithmColor}{\texttt{\bf{#1}}}%
   }}

\newcommand{\intX}[1]{\Mh{\mathrm{int}}\pth{#1}}%
\newcommand{\Line}{\Mh{\mathcalb{l}}}%
\newcommand{\LineA}{\Mh{\mathcalb{h}}}%
\newcommand{\LineB}{\Mh{\mathcalb{e}}}%

\newcommand{\ExpandOp}{\AlgorithmI{expand}\xspace}%
\newcommand{\RemoveOp}{\AlgorithmI{remove}\xspace}%

\newcommand{\avgdeg}{{d}_\mathrm{avg}}%
\newcommand{\arc}{\Mh{\gamma}}%
\newcommand{\IntSet}{\Mh{\mathcal{V}}}%
\newcommand{\vertices}{\mathcal{V}}
\newcommand{\Graph}{\Mh{G}}%
\newcommand{\Hgraph}{\Mh{H}}
\newcommand{\hedges}{\Mh{\mathcal{E}}}
\newcommand{\union}{\mathcal{U}}
\newcommand{\unionX}[1]{\union\pth{#1}}
\newcommand{\Arr}{\mathcal{A}}%
\newcommand{\ArrX}[1]{\Arr\pth{#1}}%

\newcommand{\sphereC}{\Mh{\mathbb{S}}}%
\newcommand{\simplex}{K}

\newcommand{\vDY}[2]{\Mh{V}_{\!\leq #1}\pth{#2}}
\newcommand{\depthk}{\Mh{k}}%

\newcommand{\event}{\mathcal{E}}

\newcommand{\constA}{c_1}%

\newcommand{\Edges}{\Mh{E}}%
\newcommand{\EdgesX}[1]{\Edges\pth{#1}}%
\newcommand{\nEdgesX}[1]{\Mh{m}_{#1}}%

\newcommand{\tuplesZ}[3]{L\pth{ #1, #2, #3}}%
\newcommand{\depthY}[2]{\mathrm{depth\pth{#1, #2}}}%
\newcommand{\PDSet}{\Mh{\mathcal{P}}}%

\newcommand{\Pocket}{\Mh{\Upsilon}}%

\newcommand{\nVX}[1]{\cardin{#1}}%

\providecommand{\BibLatexMode}[1]{}
\providecommand{\BibTexMode}[1]{#1}

\ifx\UseBibLatex\undefined%
  \renewcommand{\BibLatexMode}[1]{}
  \renewcommand{\BibTexMode}[1]{#1}
\else
  \renewcommand{\BibLatexMode}[1]{#1}
  \renewcommand{\BibTexMode}[1]{}
\fi

\newcommand{\UsePackage}[1]{%
  \IfFileExists{../styles/#1.sty}{%
      \usepackage{../styles/#1}%
   }{%
      \IfFileExists{./styles/#1.sty}{%
         \usepackage{styles/#1}%
      }{%
         \usepackage{#1}%
      }%
   }%
}

\BibLatexMode{%
   \usepackage[bibencoding=ascii,style=alphabetic,backend=biber]{biblatex}%
   \UsePackage{sariel_biblatex}%
}

\newcommand{\Polygon}{\Mh{\sigma}}%
\newcommand{\PolygonA}{\Mh{\pi}}%
\newcommand{\LinesX}[1]{\Mh{L}\pth{#1}}%

\newcommand{\Qin}{\Mh{Q_{\mathrm{in}}}}%
\newcommand{\Qout}{\Mh{Q_{\mathrm{out}}}}%
\newcommand{\Fin}{\Mh{F_{\mathrm{in}}}}%

\newcommand{\Fout}{\Mh{F_{\mathrm{out}}}}%

\newcommand{\ch}{\mathcal{D}}
\newcommand{\ChUp}{\Mh{\ch^+}}
\newcommand{\ChRangeY}[2]{\Mh{\ChUp}_{[#1:#2]}}%

\newcommand{\priceC}{\Mh{\varocircle}}%
\newcommand{\priceY}[2]{\priceC\pth{#1, #2}}%

\newcommand{\edge}{\Mh{e}}%
\newcommand{\edgeB}{\Mh{f}}%
\newcommand{\startX}[1]{\Mh{s}\pth{#1}}%
\newcommand{\IPSet}{\Mh{\Xi}}%
\newcommand{\IPSetZ}[3]{\IPSet\pth{#1, #2, #3}}%

\newcommand{\DotProdY}[2]{\left\langle #1, #2 \right\rangle}

\newcommand{\SaveContent}[2]{%
   \expandafter\newcommand{#1}{#2}%
}

\newcommand{\Domain}{\Mh{{D}}}%

\newcommand{\iEmpty}{\Mh{z}}%
\newcommand{\VSet}{\Mh{{V}}}%

\newcommand{\CHX}[1]{{\mathcal{CH}}\pth{#1}}

\newcommand{\normX}[1]{\| #1 \|_2}
\newcommand{\fn}{\Mh{f}}
\newcommand{\LvSetY}[2]{\mathcal{L}_{#1}\pth{#2}}
\newcommand{\opt}{\alpha}

\newcommand{\TX}{T}
\newcommand{\Tree}{\mathcal{T}}

\newcommand{\VolX}[1]{\textsf{vol}(#1)}
\newcommand{\BX}[1]{\partial #1}

\newlength{\savedparindent}

\SubmitVer{%
   \newcommand{\myparagraph}[1]{%
      \paragraph{#1}%
    }%
}%

\RegVer{%
   \newcommand{\myparagraph}[1]{%
      \paragraph*{#1}%
   }
}%
\newcommand{\ts}{\hspace{0.6pt}}

\BibLatexMode{\bibliography{resistance}}%

\title{Active-Learning a Convex Body in Low Dimensions%
   \thanks{A preliminary version appeared in ICALP
      2020~\cite{hjr-alcbl-20}.}}%
\date{\today}

\SubmitVer{%
  \author{Sariel Har-Peled\thanks{Sariel Har-Peled is supported in part by NSF AF award CCF-1907400.} \and Mitchell Jones*\thanks{* Corresponding author.} \and Saladi Rahul}
  \institute{
    S. Har-Peled \at
    Dept.~of Computer Science, University of Illinois at
    Urbana-Champaign, USA\\
    \email{sariel@illinois.edu}
    \and
    M. Jones \at
    Dept.~of Computer Science, University of Illinois at
    Urbana-Champaign, USA\\
    \email{mfjones2@illinois.edu}
    \and
    S. Rahul
    \at
    Dept.~of Computer Science and Automation, Indian Institute
    of Science, India\\
    \email{saladi@iisc.ac.in}
  }

  \authorrunning{S. Har-Peled, M. Jones, and S. Rahul}

  \journalname{Algorithmica}
}

\RegVer{%
   \author{%
      Sariel Har-Peled%
      \SarielThanks{}%
      \and%
      Mitchell Jones%
      \MitchellThanks{}%
      \and%
      Saladi Rahul%
      \RahulThanks{}%
 }%
}

\begin{document}

\maketitle

\begin{abstract}
    Consider a set $\PS \subseteq \Re^d$ of $n$ points, and a convex
    body $\body$ provided via a separation oracle. The task at hand is
    to decide for each point of $\PS$ if it is in $\body$ using the
    fewest number of oracle queries. We show
    that one can solve this problem in two and three dimensions using
    $O( \indexX{\PS} \log n)$ queries, where $\indexX{\PS}$ is the
    size of the largest subset of points of $\PS$ in convex position.
    In 2D, we provide an algorithm that efficiently generates these
    adaptive queries.

    Furthermore, we show that in two dimensions one can
    solve this problem using $O( \priceY{\PS}{\body} \log^2 n )$
    oracle queries, where $\priceY{\PS}{\body}$ is a lower bound
    on the minimum number of queries that any algorithm for this
    specific instance requires.
    Finally, we consider other variations on the problem, such
    as using the fewest number of queries to decide if $\body$
    contains all points of $\PS$.

    As an application of the above, we show that the discrete
    geometric median of a point set $P$ in $\Re^2$ can be computed
    in $O(n\log^2 n\,(\log n \log\log n + \indexX{\PS}))$ expected
    time.
\end{abstract}

\SubmitVer{\keywords{Approximation algorithms, computational geometry, separation oracles, active learning}}

\section{Introduction}

\subsection{Background}
\paragraph{Active learning.}
Active learning is a subfield of machine learning. At any
time, the learning algorithm is able to query an oracle for the label
of a particular data point. One model for active learning is the
membership query synthesis model \cite{a-qcl-87}. Here, the learner
wants to minimize the number of oracle queries, as such queries are
expensive---they usually correspond to either consulting with a
specialist, or performing an expensive computation. In this setting,
the learning algorithm is allowed to query the oracle for the label of
any data point in the instance space. See \cite{s-alls-09} for a more
in-depth survey on the various active learning models.

\begin{figure}[t]
    \phantom{}%
    \hfill%
    \includegraphics%
    [page=4,width=0.35\linewidth]%
    {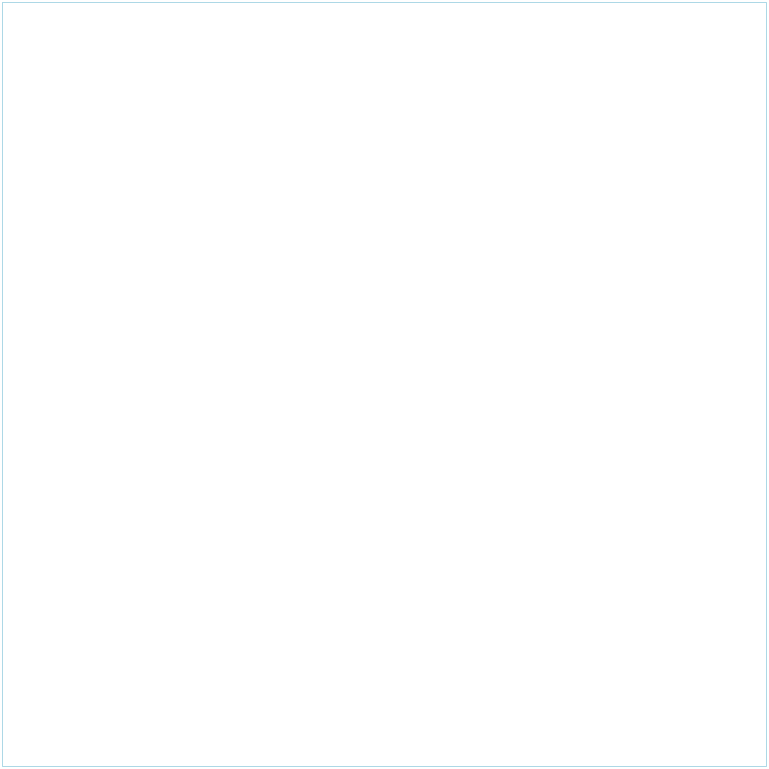}%
    \hfill%
    \includegraphics%
    [page=4,width=0.4\linewidth]%
    {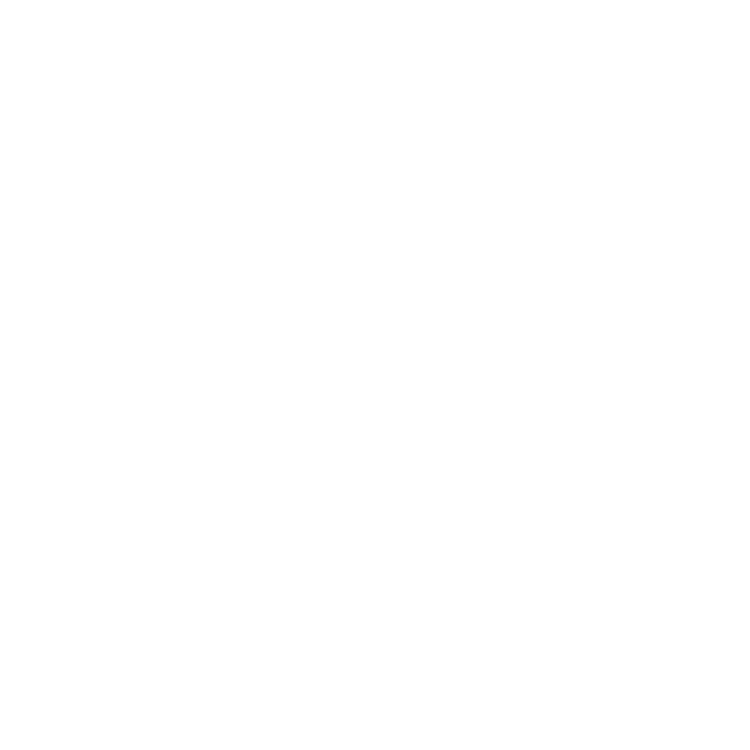}%
    \hfill\phantom{}%
    \caption{The shaded region shows the symmetric difference between
       the hypothesis and true classifier.  (I) Learning
       halfspaces. (II) Learning arbitrary convex regions.}%
    \figlab{active:learning}
\end{figure}

\myparagraph{PAC learning.}
A classical approach for learning is using random sampling, where one
gets labels for the samples (i.e., in the above setting, the oracle is
asked for the labels of all items in the random sample). PAC learning
studies the size of the sample needed.  For example, consider the
problem of learning a halfplane for $n$ points $\PS \subset \Re^2$,
given a parameter $\eps \in (0,1)$. The first stage is to take a
labeled random sample $\RSet \subseteq \PS$.  The algorithm computes
any halfplane that classifies the sample correctly (i.e., the
hypothesis).  The misclassified points lie in the symmetric difference
between the learned halfplane, and the (unknown) true halfplane, see
\figref{active:learning}. In this case, the error region is a double
wedge, and it is well known that its VC-dimension \cite{vc-ucrfe-71}
is a constant (at most eight). As such, by the $\eps$-net Theorem
\cite{hw-ensrq-87}, a sample of size $O(\eps^{-1} \log\eps^{-1})$ is
an $\eps$-net for double wedges, which implies that this random
sampling algorithm has at most $\eps n$ error.

A classical example of a hypothesis class that cannot be learned is
the set of convex regions (even in the plane). Indeed, given a set of
points $\PS$ in the plane, any sample $\RSet \subseteq \PS$ cannot
distinguish between the true region being $\CHX{\RSet}$ or
$\CHX{\PS}$. Intuitively, this is because the hypothesis space in this
case grows exponentially in the size of the sample (instead of
polynomially). See \figref{convex:example}.

\begin{figure}[h]
    \centerline{\includegraphics{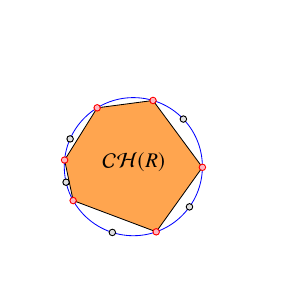}}
    \caption{}
    \figlab{convex:example}
\end{figure}

We stress that the above argument does not necessarily imply these
types of hypothesis classes are unlearnable in practice. In general,
there are other ways for learning algorithms to handle hypothesis
classes with high (or even infinite) VC-dimension (for example, using
regularization or assuming there is a large margin around the
decision boundary).

\myparagraph{Weak $\eps$-nets.}
Because $\eps$-nets for convex ranges do not exist, an interesting
direction to overcome this problem is to define weak $\eps$-nets
\cite{hw-ensrq-87}. A set of points $\RSet$ in the plane, not
necessarily a subset of $\PS$, is a \emph{weak $\eps$-net} for $\PS$
if for any convex body $\body$ containing at least $\eps n$ points of
$\PS$, it also contains a point of $\RSet$.  \Matousek and Wagner
\cite{mw-ncwe-03} gave a weak $\eps$-net construction of size
$O(\eps^{-d}(\log \eps^{-1})^{O(d^2 \log d)})$, which is doubly
exponential in the dimension. The state of the art is the recent
result of Rubin \cite{r-ibwen-18}, that shows a weak $\eps$-net
construction in the plane of size (roughly) $O(1/\eps^{3/2})$.
However, these weak $\eps$-nets cannot be used for learning such
concepts.  Indeed, the analysis above required an $\eps$-net for the
symmetric difference of two convex bodies of finite complexity, see
\figref{active:learning}.

\myparagraph{PAC learning with additional parameters.}
If one assumes the input instance obeys some additional structural
properties, then random sampling can be used. For example, suppose
that the point set $\PS$ has at most $k$ points in convex position.
For an arbitrary convex body $\body$, the convex hull
$\CHX{\PS \cap \body}$ has complexity at most $k$. Let
$\RSet \subseteq \PS$ be a random sample, and $\body'$ be the
learned classifier for $\RSet$. The region of error is the
symmetric difference between $\body$ and $\body'$.
In particular, since $k$-vertex polytopes in $\Re^d$
have VC-dimension bounded by $O(d^2 k\log k)$ \cite{k-vcdkvdp-20},
this implies that the error region also has VC-dimension at most
$O(d^2 k\log k)$. Hence if $\RSet$ is a random
sample of size $O(d^2 k\log k\eps^{-1} \log \eps^{-1})$, the
$\eps$-net Theorem \cite{hw-ensrq-87} implies that this sampling
algorithm has error at most $\eps n$. However, even for
a set of $n$ points chosen uniformly at random from the unit square
$[0,1]^2$, the expected number of points in convex position is
$O(n^{1/3})$ \cite{ab-lcc-09}. Since we want $\cardin{\RSet} < n$,
this  random sampling technique is only useful when $\eps$ is larger
than $\log^2 n /n^{2/3}$ (ignoring constants).

\smallskip\noindent%
To summarize the above discussions, random sampling on its own does
not seem powerful enough to learn arbitrary convex bodies, even if one
allows some error to be made. In this paper we focus on developing
algorithms for learning convex bodies in low dimensions, where the
algorithms are deterministic and do not make any errors.

\begin{figure}[t]
    \phantom{}%
    \hfill%
    \includegraphics[page=1,width=0.25\linewidth]{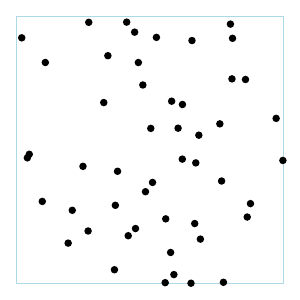}
    \hfill
    \includegraphics[page=2,width=0.25\linewidth]{figs/classify}
    \hfill
    \includegraphics[page=3,width=0.25\linewidth]{figs/classify}%
    \hfill\phantom{}%
    \caption{(I) A set of points $\PS$. (II) The unknown
    convex body $\body$. (III) Classifying all points of $\PS$ as
    either inside or outside $\body$.}%
    \figlab{classify}
\end{figure}

\subsection{Problem and motivation}

\myparagraph{The problem.}
In this paper, we consider a variation on the active learning problem,
in the membership query synthesis model. Suppose that the learner is
trying to learn an unknown convex body $\body$ in $\Re^d$.
Specifically, the learner is provided with a set $\PS$ of $n$
unlabelled points in $\Re^d$, and the task is to label each point as
either inside or outside $\body$, see \figref{classify}. For a query
$\query \in \Re^d$, the oracle either reports that $\query \in \body$,
or returns a hyperplane separating $\query$ and $\body$ (as a proof
that $\query \not\in \body$). Note that if the query is outside the
body, the oracle answer is significantly more informative than just
the label of the point. The problem is to minimize the overall number
of queries performed.

\myparagraph{Hard and easy instances.}
Note that in the worst case, an algorithm may have to query the oracle for
all input points---such a scenario happens when the input points are
in convex position, and any possible subset of the points can be the
points in the (appropriate) convex body. As such, the purpose here is
to develop algorithms that are \emph{instance sensitive}---if the
given instance is easy, they work well. If the given instance is hard,
they might deteriorate to the naive algorithm that queries all points.

Natural inputs where one can hope to do better, are when relatively
few points are in convex position. Such inputs are grid points, or
random point sets, among others.  However, there are natural instances
of the problem that are easy, despite the input having many points in
convex position. For example, consider when the convex body is a
triangle, with the input point set being $n/2$ points spread uniformly
on a tiny circle centered at the origin, while the remaining $n/2$
points are outside the convex body, spread uniformly on a circle of
radius $10$ centered at the origin, see \figref{easy}. Clearly, such a
point set can be fully classified using a sequence of a constant
number of oracle queries.  See \figref{easy:not:easy} for some related
examples.%

\begin{figure}[h]
    \centerline{\includegraphics{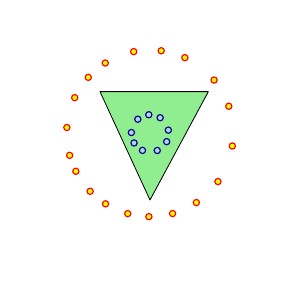}}%
    \caption{}
    \figlab{easy}
\end{figure}%

\myparagraph{Discretely optimizing convex functions.}

As an application of this particular query model, we explore the
connection between active-learning a convex body and minimizing a convex
function. Concretely, suppose we are given a set of $n$ points $\PS$ in the
plane and a convex function $\fn : \Re^2 \to \Re$ equipped with an oracle that
can evaluate $\fn$ or the derivative of $\fn$ at a given point. Our goal is to
compute the point in $\PS$ minimizing $\min_{\pnt \in \PS} \fn(\pnt)$ using
the fewest number of oracle queries (i.e., evaluations of $\fn$ or the
derivative). We discuss the result in full in \secref{applications}.

We show that there is a natural connection between the studied query model and
this problem. Namely, the level sets of a convex function are convex bodies,
and the gradient of $\fn$ can be used to construct separating lines for the
level set. Thus, developing algorithms for active-learning a convex body in the
membership query synthesis model in conjunction with the two aforementioned
facts leads to alterative methods for minimizing a convex function over a
discrete collection of points. Importantly, the running time of such algorithms
depend not only on how quickly we can evaluate $\fn$, but also on the structure
of the point set $\PS$, as we aim to develop instance sensitive algorithms.

\subsection{Additional motivation \& previous work}
\seclab{motiv:prev:work}

\myparagraph{Separation oracles.}
The use of separation oracles is a common tool in optimization
(e.g., solving exponentially large linear programs) and operations
research. It is natural to ask what other problems can be solved
efficiently when given access to this specific type of oracle.
For example, B{\'{a}}r{\'{a}}ny and F{\"{u}}redi \cite{bf-cvd-87}
study the problem of computing the volume of a convex body in
$\Re^d$ given access to a separation oracle.

\myparagraph{Other types of oracles.}
Various models of computation utilizing oracles have been previously studied within the computational geometry community.
Examples of other models include nearest-neighbor oracles (i.e.,
black-box access to nearest neighbor queries over a point set
$\PS$) \cite{hkmr-sevps-16}, proximity probes (in which given a
convex polygon $\body$ and a query $\query$, returns the distance
from $\query$ to $\body$) \cite{pavg-eppam-13}, and linear queries.
Recently, Ezra and Sharir \cite{es-nqbplha-19} gave an improved
algorithm for the problem of point location in an arrangement of
hyperplanes. Here, a \emph{linear query} consists of a point $x$ and a
hyperplane $h$, and outputs either that $x$ lies on $h$, or else the
side of $h$ contains $x$.  Alternatively, their problem can be
interpreted as querying whether or not a given point lies in a
halfspace $h^+$. Here, we study the more general problem as the
convex body can be the intersection of many halfspaces.

Furthermore, other types of active learning models
(in addition to the membership query model) have also been studied within
the learning community, see, for example, \cite{a-qcl-87}.

\myparagraph{Active learning.}
As discussed, the problem at hand can
be interpreted as active-learning a convex body in relation to a set
of points $\PS$ that need to be classified (as either inside or
outside the body), where the queries are via a separation oracle. We
are unaware of any work directly on this problem in the computational
geometry community, while there is some work in the learning
community that studies related active-learning classification
problems \cite{cal-igwal-94,gg-oalumi-07,s-alls-09,klmz-accq-17}.

For example, Kane \etal \cite{klmz-accq-17} study the problem of
actively learning halfspaces with access to \emph{comparison queries}.
Given a halfspace $h^+$ to learn, the model has two types
of queries:
\begin{compactenumi*}
  \item label queries (given $x \in \Re^d$,  is $x \in h^+$?), and
  \item comparison queries (given $x_1, x_2 \in \Re^d$, is $x_1$
  closer to the boundary of $h^+$ than $x_2$?).
\end{compactenumi*}
For example, they show that in the plane, one can classify all
points using $O(\log n)$ comparison/label queries in expectation.

\subsection{Our results}

\begin{compactenumA}
  \item We develop a greedy algorithm, for points in the plane,
  that solves the problem using $O( \indexX{\PS} \log n)$ oracle
  queries, where $\indexX{\PS}$ is the size of the largest subset of
  points of $\PS$ in convex position. See \thmref{greedy-method}.
  It is known that for a random set of $n$ points in the unit square,
  $\Ex{\indexX{\PS}} = \Theta( n^{1/3})$ \cite{ab-lcc-09}, which
  readily implies that classifying these points can be solved using
  $O( n^{1/3} \log n)$ oracle queries. A similar bound holds for the
  $\sqrt{n} \times \sqrt{n}$ grid. An animation of this algorithm is
  on YouTube \cite{j-agca-18}. We also show that this algorithm
  can be implemented efficiently, using dynamic segment trees,
  see \lemref{impl-greedy}.

  We remark that Kane \etal \cite{klmz-accq-17} develop a framework
  and randomized algorithm for learning a concept $\body$, where the
  expected number of queries depends near-linearly on a parameter they
  define as the \emph{inference dimension} \cite[Definition
  III.1]{klmz-accq-17} of the concept class.  For our problem, one can
  show that the inference dimension is $O(\indexX{\PS})$.  As a
  corollary of their framework, one can obtain a randomized algorithm
  that solves our problem where the expected number of queries is
  $O(\indexX{\PS} \log n)$, which matches the performance of our 
  deterministic algorithm. See \secref{inference} for details.

  \medskip%
  \item The above algorithm naturally extends to three dimensions,
  also using $O(\indexX{\PS} \log n)$ oracle queries. While the
  proof idea is similar to that of the algorithm in 2D, we believe
  the analysis in three dimensions is also technically interesting.
  See \thmref{greedy-method-3d}.

  \medskip%
  \item For a given point set $\PS$ and convex body $\body$, we
  define the separation price $\priceY{\PS}{\body}$ of an instance
  $(\PS, \body)$, and show that any algorithm classifying the points
  of $\PS$ in relation to $\body$ must make at least
  $\priceY{\PS}{\body}$ oracle queries (\lemref{lower:bound}).

  As an aside, we show that when $\PS$ is a set of $n$ points
  chosen uniformly at random from the unit square and $\body$
  is a (fixed) smooth convex body,
  $\Ex{\priceY{\PS}{\body}} = O(n^{1/3})$, and this bound is
  tight when $\body$ is a disk (our result also
  generalizes to higher dimensions, see \lemref{ex:sep:pr}).
  For randomly chosen points, the separation price is related
  to the expected size of the convex hull of $\PS \cap \body$,
  which is also known to be $\Theta(n^{1/3})$ \cite{b-rpcba-07}.
  We believe this result may be of independent interest, see
  \apndref{ex:sep:pr}.

  \medskip%
  \item In \secref{improved:2d} we present an improved algorithm
  for the 2D problem, and show that the number of queries made is
  $O(\priceY{\PS}{\body} \log^2 n)$. This result is $O(\log^2 n)$
  approximation to the optimal solution, see
  \thmref{improved:alg:2d}.

  \medskip%
  \item We consider the extreme scenarios of the problem: Verifying
  that all points are either inside or outside of $\body$. For each
  problem we present a $O(\log n)$ approximation algorithm to the
  optimal strategy. The results are presented in
  \secref{emptiness:2d}, see \lemref{greedy-method-empty} and
  \lemref{reverse:emptiness}.

  \medskip%
  \item \secref{applications} presents an application of the above
  results, we consider the problem of minimizing a \emph{convex} function
  $\fn : \Re^2\to \Re$ over a point set $\PS$. Specifically, the
  goal is to compute $\arg\min_{\pnt \in \PS} \fn(\pnt)$. If $\fn$
  and its derivative can be efficiently evaluated at a given query
  point, then $\fn$ can be minimized over $\PS$ using
  $O(\indexX{\PS} \log^2 n)$ queries to $\fn$ (or its derivative) in
  expectation. We refer the reader to \lemref{discrete:min}.

  Given a set of $n$ points $\PS$ in $\Re^d$, the discrete geometric
  median of $\PS$ is a point $\pnt \in \PS$ minimizing the function
  $\sum_{\pntq \in \PS} \normX{\pnt - \pntq}$.  As a corollary of
  \lemref{discrete:min}, we obtain an algorithm for computing the
  discrete geometric median for $n$ points in the plane. The algorithm
  runs in $O(n\log^2 n \cdot (\log n \log\log n + \indexX{\PS}))$
  expected time. See \lemref{discrete:med}. In particular, if $\PS$
  is a set of $n$ points chosen uniformly at random from the unit
  square, it is known that $\Ex{\indexX{\PS}} = \Theta(n^{1/3})$
  \cite{ab-lcc-09} and hence the discrete geometric median can be
  computed in $O(n^{4/3} \log^2 n)$ expected time.

  While there has been ample work on approximating the geometric
  median (recently, Cohen \etal \cite{clmps-gmnlt-16} gave a
  $(1+\eps)$-approximation algorithm to the geometric median in
  $O(dn \log^3 (1/\eps))$ time), we are unaware
  of any \emph{exact} sub-quadratic algorithm for the discrete
  case even in the plane.
\end{compactenumA}

Throughout this paper, the model of computation we have assumed is
unit-cost real RAM.

\section{The greedy algorithm in two and three dimensions}
\seclab{greedy-2d-3d}

\subsection{Preliminaries}

For a set of points $\PS \subseteq \Re^2$, let $\CHX{\PS}$ denote
the convex hull of $\PS$.  Given a convex body
$\body \subseteq \Re^d$, two points
$\pnt, \pntx \in \Re^d \setminus \intX{\body}$ are \emphi{mutually visible},
if the segment $\pnt \pntx$ does not intersect $\intX{\body}$, where
$\intX{\body}$ is the interior of $\body$.
We also use the notation
$\PS \cap \body = \Set{\pnt \in \PS}{\pnt \in \body}$.

For a point set $\PS \subseteq \Re^d$, a \emphi{centerpoint} of
$\PS$ is a point $\cpnt \in \Re^d$, such that for any closed
halfspace $h^+$ containing $\cpnt$, we have
$\cardin{h^+ \cap \PS} \geq \cardin{\PS}/(d+1)$.  A centerpoint
always exists, and it can be computed exactly in
$O(n^{d-1} + n \log n)$ time \cite{c-oramt-04}.

Let $\body$ be a convex body in $\Re^d$ and $\pb \in \Re^d$ be a point
such that $\pb$ lies outside $\body$. A hyperplane $h$
\emphi{separates} $\pb$ from $\body$ if $\pb$ lies in the \emph{closed}
halfspace $h^+$ bounded by $h$, and $\body$ is contained in the \emph{open}
halfspace $h^-$ bounded by $h$. This definition allows the separating
hyperplane to contain the point $\pb$, and will simplify the descriptions
of the algorithms.

\subsection{The greedy algorithm in 2D}

Given a set of points $\PS$ in $\Re^2$ and a convex body $\body$ specified
via a separation oracle, recall that the problem is to classify, for all the
points of $\PS$, whether or not they are in $\body$, using the fewest
oracle queries possible. We define some operations the algorithm will use
before stating the algorithm in full. Finally, we analyze the the number of
queries the algorithm performs.

\subsubsection{Operations}
\seclab{operations}

Initially, the algorithm copies $\PS$ into a set $\USet$ of
unclassified points.  The algorithm is going to maintain an inner
approximation $\inapx \subseteq \body$.  There are two types of
updates (\figref{operations} illustrates the two operations):
\begin{compactenumA}
    \smallskip%
    \item \ExpandOp{}$(\pnt)$: Given a point
    $\pnt \in \body \setminus \inapx$, the algorithm is going to:
    \begin{compactenumi}
        \item Update the inner approximation:
        $\inapx \leftarrow \CHX{\inapx \cup \brc{\pnt}}$.

        \item Remove (and mark) newly covered points:
        $\USet \leftarrow \USet \setminus \inapx$.
    \end{compactenumi}

    \medskip%
    \item \RemoveOp{}$(\Line)$: Given a closed halfplane $\Line^+$
    such that $\intX{\body} \cap \Line^+ = \emptyset$, the algorithm
    marks all the points of $\USet_\Line = \USet \cap \intX{\Line^+}$
    as being outside $\body$, and sets
    $\USet \leftarrow \USet \setminus \USet_\Line$.
\end{compactenumA}

\begin{figure}[h]
    \phantom{}\hfill%
    \includegraphics%
    [page=1,width=0.4\linewidth]%
    {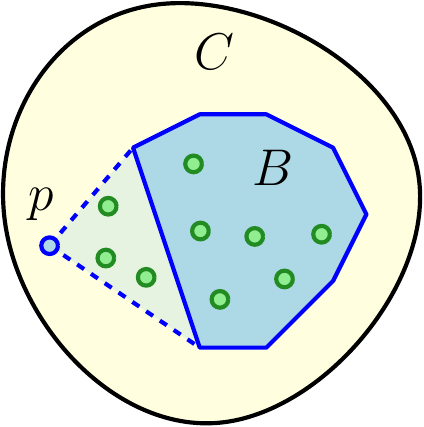}%
    \hfill%
    \includegraphics%
    [page=2,width=0.4\linewidth]%
    {figs/operations_expand}%
    \hfill\phantom{}%

    \caption{(I) Performing $\ExpandOp{}(\pnt)$, and marking points
       inside $\body$. (II) Performing $\RemoveOp{}(\Line)$, and
       marking points outside $\body$.}
    \figlab{operations}
\end{figure}

\subsubsection{The algorithm}
\seclab{round}%

The algorithm repeatedly performs rounds, as described next, until the
set of unclassified points is empty.

At every round, if the inner approximation $\inapx$ is empty, then the
algorithm sets $\USet^+ = \USet$. Otherwise, the algorithm picks a
line $\Line$ that is tangent to $\inapx$ with the largest number of
points of $\USet$ on the other side of $\Line$ than $\inapx$.  Let
$\Line^-$ and $\Line^+$ be the two closed halfspace bounded by
$\Line$, where $\inapx \subseteq \Line^-$. The algorithm computes the
point set $\USet^+ = \USet \cap \Line^+$.  We have two cases:
\medskip%
\begin{compactenumA}[label=\Alph*.]
    \item Suppose $\cardin{\USet^+}$ is of constant size. The
    algorithm queries the oracle for the status of each of these
    points.  For every point $\pnt \in \USet^+$, such that $\pnt \in
    \body$, the algorithm performs \ExpandOp{}$(\pnt)$. Otherwise, the
    oracle returned a separating line $\Line$, and the algorithm calls
    \RemoveOp{}$(\Line^+)$.

    \medskip%
    \item Otherwise, $\cardin{\USet+}$ does not have constant size.
    The algorithm computes a centerpoint
    $\cpnt \in \Re^2$ for $\USet^+$, and asks the oracle for the
    status of $\cpnt$. There are two possibilities:
    \begin{compactenumA}[label*=\Roman*.]
        \item If $\cpnt \in \body$, then the algorithm performs
        \ExpandOp{}$(\cpnt)$.

        \item If $\cpnt \notin \body$, then the oracle returned a
        separating line $\LineA$, and the algorithm performs
        \RemoveOp{}$(\LineA)$.
    \end{compactenumA}
\end{compactenumA}

\subsubsection{Analysis}

Let $\inapx_i$ be the inner approximation at the start of the $i$\th
iteration, and let $\iEmpty$ be the first index where $\inapx_\iEmpty$
is not an empty set. Similarly, let $\USet_i$ be the set of
unclassified points at the start of the $i$\th iteration, where
initially $\USet_1 = \USet$.

\begin{lemma}
    \lemlab{iterations:empty:inapx}%
    The number of (initial) iterations in which the inner
    approximation is empty is $\iEmpty = O(\log n)$.
\end{lemma}
\begin{proof}
    As soon as the oracle returns a point that is in $\body$, the
    inner approximation is no longer empty. As such, we need to bound
    the initial number of iterations where the oracle returns that the
    query point is outside $\body$.  Let $f_i = \cardin{\USet_i}$, and
    note that $\USet_1 = \PS$ and $f_1 = \cardin{\PS} = n$.  Let
    $\cpnt_i$ be the centerpoint of $\USet_i$, which is the query point
    in the $i$\th iteration ($\cpnt_i$ is outside $\body$). As such,
    the line separating $\cpnt_i$ from $\body$, returned by the
    oracle, has at least $f_i/3$ points of $\USet_i$ on the same side
    as $\cpnt_i$, by the centerpoint property.  All of these points get
    labeled in this iteration, and it follows that
    $f_{i+1} \leq (2/3) f_i$, which readily implies the claim, since
    $f_{\iEmpty} < 1$, for $\iEmpty = \ceil{\log_{3/2} n} +1$.
\end{proof}

\begin{definition}[Visibility graph]
    \deflab{visi:graph}%
    Consider the graph $\Graph_i$ over $\USet_i$, where two points
    $\pnt, \pntB \in \USet_i$ are connected $\iff$ the segment
    $\pnt \pntB$ does not intersect the interior of $\inapx_i$.
\end{definition}

\begin{figure}[h]
      \phantom{}\hfill%
      \begin{minipage}{0.3\linewidth}
          \includegraphics[page=1,width=0.99\linewidth]{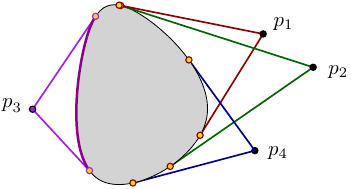}
      \end{minipage}
      \hfill%
      \begin{minipage}{0.28\linewidth}
          \includegraphics%
          [page=2,width=0.99\linewidth]{figs/intervals}%
      \end{minipage}
      \hfill%
      \begin{minipage}{0.28\linewidth}
          \includegraphics%
          [page=3,width=0.99\linewidth]%
          {figs/intervals}%
      \end{minipage}
      \hfill\phantom{}%

    \caption{Four points and a convex body with their associated circular
    intervals.}
    \figlab{intervals}
\end{figure}

\myparagraph{The visibility graph as an interval graph.}
For a point $\pnt \in \USet_i$, let $\IY{i}{\pnt}$ be the set of all
directions $v$ (i.e., vectors of length $1$) such that there is a line
perpendicular to $v$ that separates $\pnt$ from $\inapx_i$. Formally,
a line $\Line$ separates $\pnt$ from $\inapx_i$, if the interior of
$\inapx_i$ is on one side of $\Line$ and $\pnt$ is on the (closed)
other side of $\Line$ (if $\pnt \in \Line$, the line is still
considered to separate the two). Clearly, $\IY{i}{\pnt}$ is a circular
interval on the unit circle. See \figref{intervals}.
The resulting set of intervals is
$\IntSet_i = \Set{\IY{i}{\pnt}}{\pnt \in \USet_i}$. It is easy to
verify that the intersection graph of $\IntSet_i$ is $\Graph_i$.
Throughout the execution of the algorithm, the inner approximation
$\inapx_i$ grows monotonically, this in turn implies that the
visibility intervals shrink over time; that is,
$\IY{i}{\pnt} \subseteq \IY{i-1}{\pnt}$, for all $\pnt \in \PS$ and
$i$.  Intuitively, in each round, either many edges from $\Graph_i$
are removed (because intervals had shrunk and they no longer
intersect), or many vertices are removed (i.e., the associated points
are classified).

\begin{definition}
    Given a set $\IntSet$ of objects (e.g., intervals) in a domain
    $\Domain$ (e.g., unit circle), the \emphi{depth} of a point
    $\pnt \in D$, is the number of objects in $\IntSet$ that contain
    $\pnt$. Let $\depthX{\IntSet}$ be the maximum depth of any point
    in $\Domain$.
\end{definition}

When it is clear, we use $\depthX{\Graph}$ to denote
$\depthX{\IntSet}$, where $\Graph = (\IntSet, \Edges)$ is the
intersection graph of the intervals $\IntSet$ as defined above.
Throughout, we commonly refer to $\Graph$ as the
\emphi{intersection graph}.

First, we bound the number of edges in this visibility graph
$\Graph$ and then argue that in each iteration, either many edges
of $\Graph$ are discarded or vertices are removed (as they are
classified).

\begin{lemma}
    \lemlab{int-graphs}%
    Let $\IntSet$ be a set of $n$ intervals on the unit circle, and
    let $\Graph = (\IntSet, \Edges)$ be the associated intersection
    graph.  Then $\cardin{\Edges} =O(\indep \clique^2)$, where
    $\clique= \depthX{\IntSet}$ and $\indep = \indep(\Graph)$ is the
    size of the largest independent set in $\Graph$. Furthermore,
    the upper bound on $\cardin{\Edges}$ is tight.
\end{lemma}
\begin{proof}
    Let $J$ be the largest independent set of intervals in
    $\Graph$. The intervals of $J$ divide the circle into
    $2\cardin{J}$ (atomic) circular arcs.  Consider such an arc
    $\arc$, and let $K(\arc)$ be the set of all intervals of $\IntSet$
    that are fully contained in $\arc$. All the intervals of $K(\arc)$
    are pairwise intersecting, as otherwise one could increase the
    size of the independent set. As such, all the intervals of
    $K(\arc)$ must contain a common intersection point. It follows
    that $\cardin{K(\arc)} \leq \clique$.

    Let $K'(\arc)$ be the set of all intervals intersecting
    $\arc$. This set might contain up to $2\clique$ additional
    intervals (that are not contained in $\arc$), as each such
    additional interval must contain at least one of the endpoints of
    $\arc$. Namely, $\cardin{K'(\arc)} \leq 3 \clique$. In particular,
    any two intervals intersecting inside $\arc$ both belong to
    $K'(\arc)$. As such, the total number of edges contributed by
    $K'(\arc)$ to $\Graph$ is at most
    $\binom{3\clique}{2} = O(\clique^2)$.  Since there are at most
    $2 \indep$ arcs under consideration, the total number of
    edges in $\Graph$ is bounded by $O(\indep \clique^2)$, which
    implies the claim.

    The lower bound is easy to see by taking an independent set of
    intervals of size $\indep$, and replicating every interval
    $\clique$ times.
\end{proof}

\begin{lemma}
    \lemlab{segments:intersect}%
    Let $\PS$ be a set of $n$ points in the plane lying above the
    $x$-axis, $\cpnt$ be a centerpoint of $\PS$, and
    $\SegSet = \binom{\PS}{2}$ be set of all segments induced by
    $\PS$. Next, consider any point $\pntB$ on the $x$-axis. Then,
    the segment $\cpnt \pntB$ intersects at least $n^2/36$ segments of
    $\SegSet$.
\end{lemma}
\begin{proof}
    If the segment $\cpnt \pntB$ intersects the segment
    $\pnt_1 \pnt_2$, for $\pnt_1, \pnt_2 \in \PS$, then we consider
    $\pnt_1$ and $\pnt_2$ to no longer be mutually visible.
    It suffices to lower bound the number of pairs of points
    that lose mutual visibility of each other.

    \begin{figure}[h]
        \centerline{\includegraphics{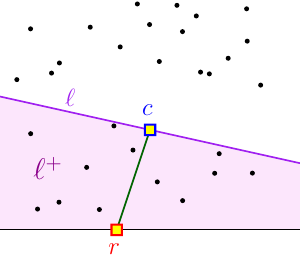}}%
        \caption{}
        \figlab{center:point}
    \end{figure}

    Consider a line $\Line$ passing through the point $\cpnt$, see
    \figref{center:point}..  Let $\Line^+$ be the closed halfspace
    bounded by $\Line$ containing $\pntB$. Note that
    $\cardin{\PS \cap \Line^+} \geq n/3$, since $\cpnt$ is a
    centerpoint of $\PS$, and $\cpnt \in \Line$. Rotate $\Line$ around
    $\cpnt$ until there are $\geq n/6$ points on each side of
    $\pntB\cpnt$ in the halfspace $\Line^+$. To see why this rotation
    of $\Line$ exists, observe that the two halfspaces bounded by the
    line spanning $\pntB\cpnt$, have zero points on one side, and at
    least $n/3$ points on the other side --- a continuous rotation of
    $\Line$ between these two extremes, implies the desired property.

    Observe that points in $\Line^+$ and on opposite sides of the
    segment $\cpnt\pntB$ cannot see each other, as the segment
    connecting them must intersect $\cpnt\pntB$.  Consequently, the
    number of induced segments that $\cpnt\pntB$ intersects is at
    least $n^2/36$.
\end{proof}

For a graph $\Graph$, we let $\EdgesX{\Graph}$ denote the set of edges in $G$,
and let $\cardin{\EdgesX{G}}$ denote the number of edges in $\Graph$.

\begin{lemma}
    \lemlab{depth:reduce}%
    Let $\Graph_i$ be the intersection graph, in the beginning of the
    $i$\th iteration, and let
    $\nEdgesX{i} = \cardin{\EdgesX{\Graph_{i}}}$. After the $i$\th
    iteration of the greedy algorithm, we have
    $\nEdgesX{i+1} \leq \nEdgesX{i} - \clique^2/36$, where
    $\clique = \depthX{\Graph_i}$.
\end{lemma}
\begin{proof}
    Recall that in the algorithm $\USet^+ = \USet_i \cap \Line^+$ is
    the current set of unclassified points and $\Line$ is the line
    tangent to $\inapx_i$, where $\Line^+$ is the closed halfspace
    that avoids the interior of $\inapx_i$ and contains the largest
    number of unlabeled points of $\USet_i$.  We have that
    $\clique = \cardin{\USet^+}$.

    If a \RemoveOp{} operation was performed in the $i$\th iteration,
    then the number of points of $\USet^+$ that are discarded is at
    least $\clique/3$. In this case, the oracle returned a separating
    line $\LineA$ between a centerpoint $\cpnt$ of $\USet^+$ and the
    inner approximation. For the halfspace $\LineA^+$ containing
    $\cpnt$, we have
    $t_i = \cardin{\USet^+ \cap \LineA^+} \geq \cardin{\USet^+}/3 \geq
    \clique/3$. Furthermore, all the points of $\USet^+$ are pairwise
    mutually visible (in relation to the inner approximation
    $\inapx_i$). Namely,
    $$
        \nEdgesX{i+1}%
        =%
        \cardin{\EdgesX{\Graph_{i} - (\USet^+ \cap \LineA^+)}}%
        \leq%
        \nEdgesX{i} - \binom{t_i}{2}%
        \leq%
        \nEdgesX{i} - \clique^2 /36.
    $$

    If an \ExpandOp{} operation was performed, the centerpoint $\cpnt$
    of $\USet^+$ is added to the current inner approximation
    $\inapx_i$. Let $\pntB$ be a point in $\Line \cap \inapx_i$, and
    let $\cpnt_i$ be the centerpoint of $\USet_i$ computed by the
    algorithm.  By \lemref{segments:intersect} applied to
    $\pntB, \cpnt$ and $\USet^+$, we have that at least $\clique^2/36$
    pairs of points of $\USet^+$ are no longer mutually visible to each
    other in relation to $\inapx_{i+1}$. We conclude, that at least
    $\clique^2/36$ edges of $\Graph_i$ are no longer present in
    $\Graph_{i+1}$.
\end{proof}

\begin{definition}
    \deflab{index}%
    A subset of points $X \subseteq \PS \subseteq \Re^2$ are in
    \emphi{convex position}, if all the points of $X$ are vertices of
    $\CHX{X}$ (note that a point in the middle of an edge is not
    considered to be a vertex). The \emphi{index} of $\PS$, denoted
    by $\indexX{\PS}$, is the cardinality of the largest subset of
    $\PS$ of points that are in convex position.
\end{definition}

\begin{theorem}
    \thmlab{greedy-method}%
    Let $\body$ be a convex body provided via a separation oracle, and
    let $\PS$ be a set of $n$ points in the plane. The greedy
    classification algorithm performs
    $O\bigl((\indexX{\PS}+1) \log n\bigr)$ oracle queries. The
    algorithm correctly identifies all points in $\PS \cap \body$
    and $\PS \setminus \body$.
\end{theorem}
\begin{proof}
    By \lemref{iterations:empty:inapx}, the number of iterations (and
    also queries) in which the inner approximation is empty is
    $O(\log n)$, and let $\iEmpty = O(\log n)$ be the first iteration
    such that the inner approximation is not empty.  It suffices to
    bound the number of queries made by the algorithm after the inner
    approximation becomes non-empty.

    For $i \geq \iEmpty$, let $\Graph_i = (\USet_i, \Edges_i)$ denote
    the visibility graph of the remaining unclassified points
    $\USet_i$ in the beginning of the $i$\th iteration.  Any
    independent set in $\Graph_i$ corresponds to a set of points
    $X \subseteq \PS$ that do not see each other due to the presence
    of the inner approximation $\inapx_i$. That is, $X$ is in
    convex position, and furthermore $\cardin{X} \leq \indexX{\PS}$.

    For $0 \leq t \leq n$, let $\startX{t}$ be the first iteration
    $i$, such that $\depthX{ \Graph_i} \leq t$. Since the depth of
    $\Graph_i$ is a monotone decreasing function, this quantity is
    well defined.  An \emphi{epoch} is a range of iterations between
    $s(t)$ and $s(t/2)$, for any parameter $t$. We claim that an epoch
    lasts $O( \indexX{ \PS})$ iterations (and every iteration issues
    only one oracle query). Since there are only $O( \log n)$
    (non-overlapping) epochs till the algorithm terminates, as the
    depth becomes zero, this implies the claim.

    So consider such an epoch starting at $i = \startX{t}$. We have
    $m = \nEdgesX{i} = \cardin{\EdgesX{\Graph_i}} = O( \indexX{\PS}
    t^2)$, by \lemref{int-graphs}, since $\indexX{\PS}$ is an
    upper bound on the size of the largest independent set in
    $\Graph_i$.  By \lemref{depth:reduce}, as long as the depth of the
    intervals is at least $t/2$, the number of edges removed from the
    graph at each iteration, during this epoch, is at least
    $\Omega(t^2)$. As such, the algorithm performs at most
    $O(m_i/t^2) = O( \indexX{\PS} )$ iterations in this epoch, till
    the maximum depth drops to $t/2$.
\end{proof}

\subsubsection{Implementing the greedy algorithm}

With the use of dynamic segment trees \cite{mn-dfc-90} we show
that the greedy classification algorithm can be implemented
efficiently.

\begin{lemma}
    \lemlab{impl-greedy}
    Let $\body$ be a convex body provided via a separation oracle, and
    let $\PS$ be a set of $n$ points in the plane. If an oracle query
    costs time $\TX$, then the greedy algorithm can be implemented in
    $O\bigl(n\log^2 n\log\log n + \TX\cdot\indexX{\PS} \log n\bigr)$
    expected time.
\end{lemma}
\begin{proof}
    The algorithm follows the proof of \thmref{greedy-method}. We
    focus on efficiently implementing the algorithm once inner
    approximation is no longer empty.  Let $\USet \subseteq \PS$ be
    the subset of unclassified points. By binary searching on the
    vertices of the inner approximation $\inapx$, we can compute the
    collection of visibility intervals $\IntSet$ for all points in
    $\USet$ in $O(\cardin{\USet}\log m) = O(n\log n)$ time (recall
    that $\IntSet$ is a collection of circular intervals on the unit
    circle). We store these intervals in a dynamic segment tree
    $\Tree$ with the modification that each node $v$ in $\Tree$ stores
    the maximum depth over all intervals contained in the subtree
    rooted at $v$.  Note that $\Tree$ can be made fully dynamic to
    support updates in $O(\log n \log \log n)$ time \cite{mn-dfc-90}.

    An iteration of the greedy algorithm proceeds as follows. Start by
    collecting all points $\USet^+ \subseteq \USet$ realizing the
    maximum depth using $\Tree$. When $t = \cardin{\USet^+}$, this
    step can be done in $O(\log n + t)$ time by traversing $\Tree$.
    We compute the centerpoint of $\USet^+$ in $O(t \log t)$ expected
    time \cite{c-oramt-04} and query the oracle using this
    centerpoint. Either points of $\USet$ are classified (and we
    delete their associated intervals from $\Tree$) or we improve the
    inner approximation. The inner approximation (which is the convex
    hull of query points inside the convex body $\body$) can be
    maintained in an online fashion with insert time $O(\log n)$
    \cite[Chapter 3]{ps-cg-85}.  When the inner approximation expands,
    the points of $\USet^+$ have their intervals shrink. As such, we
    recompute $\IX{\pnt}$ for each $\pnt \in \USet^+$ and reinsert
    $\IX{\pnt}$ into $\Tree$.

    As defined in the proof of \thmref{greedy-method}, an epoch is the
    subset of iterations in which the maximum depth is in the range
    $[t/2, t]$, for some integer $t$. During such an epoch, we make
    two claims:
    \smallskip%
    \begin{compactenumi}
      \item there are $\sigma = O(n)$ updates to $\Tree$, and
      \item the greedy algorithm performs $O(n/t)$
      centerpoint calculations on sets of size $O(t)$.
    \end{compactenumi}
    \smallskip%

    Both of these claims imply that a single epoch of the greedy
    algorithm can be implemented in expected time
    $O(\sigma \log n \log\log n + n\log n + \TX\cdot \indexX{\PS})$.
    As there are $O(\log n)$ epochs, the algorithm can be
    implemented in expected time
    $O(n \log^2 n \log\log n + \TX\cdot \indexX{\PS}\log n)$.

    We now prove the first claim. Recall that we have a collection of
    intervals $\IntSet$ lying on the circle of directions. Partition
    the circle into $k$ atomic arcs, where each arc contains $t/10$
    endpoints of intervals in $\IntSet$. Note that $k = 20n/t =
    O(n/t)$. For each circular arc $\arc$, let
    $\IntSet_\arc \subseteq \IntSet$ be the set of intervals
    intersecting $\arc$. As the maximum depth is bounded by $t$, we
    have that $\cardin{\IntSet_\arc} \leq t + t/10 = 1.1t$.  In
    particular, if $\Graph[\IntSet_\arc]$ is the induced subgraph of
    the intersection graph $\Graph$, then $\Graph[\IntSet_\arc]$ has
    at most $\binom{\cardin{\IntSet_\arc}}{2} = O(t^2)$ edges.

    In each iteration, the greedy algorithm chooses a point in an
    arc $\arc$ (we say that $\arc$ is \emph{hit}) and edges
    are only deleted from $\Graph[\IntSet_\arc]$.
    The key observation is that an arc $\arc$ can only be hit
    $O(1)$ times before all points of $\arc$ have depth below
    $t/2$, implying that it will not be hit again until the next
    epoch. Indeed, each time $\arc$ is hit, the number of edges
    in the induced subgraph $\Graph[\IntSet_\arc]$ drops
    by a constant factor (\lemref{depth:reduce}). Additionally,
    when $\Graph[\IntSet_\arc]$ has less than
    $\binom{t/2}{2}$ edges then any point on $\arc$ has
    depth less than $t/2$. These two facts imply that an arc
    is hit $O(1)$ times.

    When an arc is hit, we must reinsert
    $\cardin{\IntSet_\arc} = O(t)$ intervals into $\Tree$. In
    particular, over a single epoch, the total number of hits
    over all arcs is bounded by $O(k)$. As such,
    $\sigma = O(kt) = O(n)$.

    For the second claim, each time an arc is hit, a single
    centerpoint calculation is performed. Since each arc
    has depth at most $t$ and is hit a constant number
    of times, there are $O(k) = O(n/t)$ such
    centerpoint calculations in a single epoch, each costing
    expected time $O(t\log t)$.
\end{proof}

In \secref{applications} we present an application of the
greedy classification algorithm. Namely, we present an efficient
algorithm for computing the discrete geometric median of a point set
(\lemref{discrete:med}).

\subsubsection{The inference dimension and an alternative
   algorithm}
\seclab{inference}

Kane \etal \cite{klmz-accq-17} define the notion of \emph{inference
   dimension}, which in our context is the minimum number of queries
needed to classify all points.

\begin{figure}[h]
    \centerline{\includegraphics{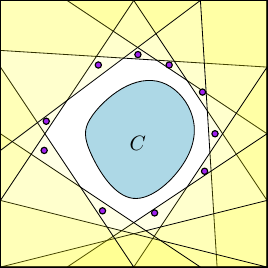}}
    \caption{The minimal external set must be convex.}
    \figlab{outer:conex}
\end{figure}

\begin{lemma}
    Let $\body$ be a convex body provided via a
    separation oracle, and let $\PS$ be a set of $n$ points in
    the plane. There is a set of $2\indexX{\PS}$ oracle queries whose
    answers can be used to classify all points of $\PS$ correctly.
\end{lemma}
\begin{proof}
    We put the at most $\indexX{\PS}$ vertices of
    $\CHX{ \PS \cap \body}$ into a query set. Querying these points is
    enough to label correctly all points inside the body
    $\body$. As for the points of $\PS$ outside $\body$, let
    $\PSA \subseteq \PS \setminus \body$ be the minimum size subset
    such that querying these points correctly labels all
    points outside $\body$. Each point $\pa \in \PSA$ is associated with a
    halfspace $h^+_\pa$ that contains $\body$. Let $H$ be this set
    of halfspaces. Observe that for any point $\pa \in \PSA$, there is
    a point $\mathrm{witness}(\pa)\in \PS \setminus \body$ for which 
    $h^+_\pa$ does not contain $\mathrm{witness}(\pa)$ (as
    otherwise, $\pa$ can be removed from $\PSA$). Let
    $\PSB = \Set{\mathrm{witness}(\pa)}{\pa \in \PSA}$.  The points of
    $U$ are in the faces of the arrangement $\ArrX{H}$ that are
    adjacent to the face $\cap_{\pa \in \PSA} h^+_\pa$, see
    \figref{outer:conex}.

    Since each point of $\PSB$ is separable by a line from the
    remaining points of $\PSB$, it follows that $\PSB$ is convex. As
    such, $\cardin{\PSA} = \cardin{\PSB} \leq \indexX{\PS}$
    which implies the result.
\end{proof}

The above lemma implies that the inference dimension of $\PS$ is
$2\indexX{\PS}$. Plugging this into the algorithm of Kane \etal
\cite{klmz-accq-17} results in an algorithm that labels all points correctly
and performs the same number of queries as \thmref{greedy-method} in
expectation.  The advantage of \thmref{greedy-method} is that it does not
require knowing the value of $\indexX{\PS}$ in advance. However, one could
perform an exponential search for a tight upper bound on $\indexX{\PS}$, and
still use the algorithm of Kane \etal\cite{klmz-accq-17}. We leave the question
of experimentally comparing the two algorithms as an open problem for future
research.

\paragraph{Sketch of the algorithm of \cite{klmz-accq-17}.}
The algorithm of Kane \etal \cite{klmz-accq-17} specialized for our case works
as follows. Start by randomly picking a sample of size $O( \indexX{\PS})$ and
query the oracle with each of these points. Next, stream the unlabeled points
through the computed regions, leaving only the points that are yet to be
labeled.  The algorithm repeats this process $O( \log n )$ times, in each
iteration working on the remaining unlabeled points. By proving that in
expectation at least half of the points are being labeled at each round, it 
follows that $O(\log n)$ iterations suffice.

\subsection{The greedy algorithm in 3D}
\seclab{greedy:3d}

Consider the 3D variant of the 2D problem: Given a set of points
$\PS$ in $\Re^3$ and a convex body $\body$ specified via a
separation oracle, the task at hand is to classify, for all the points
of $\PS$, whether or not they are in $\body$, using the fewest
oracle queries possible.

The greedy algorithm naturally extends, where at each iteration $i$ a
plane $\LineB_i$ is chosen that is tangent to the current inner
approximation $\inapx_i$, such that it's closed halfspace (which
avoids the interior of $\inapx_i$) contains the largest number of
unclassified points from the set $\USet_i$. If the queried centerpoint
is outside, the oracle returns a separating plane and as such points
can be discarded by the \RemoveOp{} operation. Similarly, if the
centerpoint is reported inside, then the algorithm calls the
\ExpandOp{} and updates the 3D inner approximation $\inapx_i$. %

\subsubsection{Analysis}

Following the analysis of the greedy algorithm in 2D, we
(conceptually) maintain the following set of objects: For a point
$\pnt \in \USet_i$, let $\pdisk_i(\pnt)$ be the set of all unit length
directions $v \in \Re^3$ such that a plane perpendicular to $v$
separates $\pnt$ from $\inapx_i$. Let
$\PDSet_i = \Set{\pdisk_i(\pnt)}{\pnt \in \USet_i}$. A set of objects
form a collection of \emphi{pseudo-disks} if the boundary of every
pair of them intersect at most twice. The following claim shows that
$\PDSet_i$ is a collection of pseudo-disks on $\sphereC$, where
$\sphereC$ is the sphere of radius one centered at the origin.

\begin{lemma}
    The set
    $\PDSet_i = \Set{\pdisk_i(\pnt) \subseteq \sphereC}{\pnt \in
       \USet_i}$ is a collection of pseudo-disks.
\end{lemma}
\begin{proof}
    Fix two points $\pnt, \pntB \in \USet_i$ such that the boundaries
    of $\pdisk_i(\pnt)$ and $\pdisk_i(\pntB)$ intersect on
    $\sphereC$. Let $\Line$ be the line in $\Re^3$ passing through
    $\pnt$ and $\pntB$. Consider any plane $\LineB$ such that $\Line$
    lies on $\LineB$. Since $\Line$ is fixed, $\LineB$ has one degree
    of freedom. Conceptually rotate $\LineB$ until becomes tangent to
    $\inapx_i$ at point $\pntC'$. The direction of the normal to this
    tangent plane, is a point in
    $X = \partial\pdisk_i(\pnt) \cap \partial\pdisk_i(\pntB)$. Note
    that this works also in the other direction --- any point in $X$
    corresponds to a tangent plane passing through $\Line$.  The
    family of planes passing through $\Line$ has only two tangent
    planes to $\body$. It follows that $\cardin{X}=2$. As such, any
    two regions in $\PDSet_i$ intersect as pseudo-disks.
\end{proof}

We need the following two classical results that follows from the
Clarkson-Shor \cite{cs-arscg-89} technique.

\begin{lemma}
    \lemlab{num:vertices:depth}%
   Let $\PDSet$ be a collection of $n$ pseudo-disks, and let
   $\vDY{\depthk}{\Arr}$ be the set of all vertices of depth at most
   $\depthk$ in the arrangement $\Arr = \ArrX{\PDSet}$.  Then
   $\cardin{\vDY{\depthk}{\Arr}} = O(n\depthk)$.%
\end{lemma}
\begin{proof}
    Let $\SSet \subseteq \vertices$ be a random sample where each
    pseudo-disk is independently placed into $\SSet$ with probability
    $1/\depthk$. For each $\pnt \in \vDY{\depthk}{\Arr}$, let
    $\event_\pnt$ be the event that $\pnt$ is a vertex in the union
    $\unionX{\SSet}$ of this random subset of pseudo-disks.  The
    probability that $\pnt$ is part of the union is at least the
    probability that both pseudo-disks defining $\pnt$ in $\Arr$ are
    sampled into $\SSet$ and the remaining $\depthk-2$ objects
    containing $\pnt$ are not in $\SSet$. Thus,
    \begin{align*}
      \Prob{\event_\pnt}
      \geq \frac{1}{\depthk^2} \pth{1 - \frac{1}{\depthk}}^{\depthk}
      \geq \frac{1}{e^2 \depthk^2},
    \end{align*}
    since $1 - 1/x \geq e^{-2/x}$ for $x \geq 2$. If
    $\cardin{\unionX{\SSet}}$ denotes the number of vertices on the
    boundary of the union, then linearity of expectations imply
    $\Ex{\cardin{\unionX{\SSet}}} \geq
    \cardin{\vDY{\depthk}{\Arr}}/(e^2 \depthk^2)$. On the other hand,
    it is well known the union complexity of a collection of $n$
    pseudo-disks is $O(n)$ \cite{klps-ujrcf-86}. Therefore,
    $\Ex{\cardin{\unionX{\SSet}}} \leq \Ex{c \cardin{\SSet}} \leq
    cn/\depthk$, for some appropriate constant $c$. Putting both
    bounds on $\Ex{\cardin{\unionX{\SSet}}}$ together, it follows that
    $cn/\depthk \geq \cardin{\vDY{\depthk}{\Arr}}/(e^2 \depthk^2) \iff
    \cardin{\vDY{\depthk}{\Arr}} = O(n\depthk)$.
\end{proof}

\begin{lemma}
    \lemlab{num:edges:depth}%
   Let $\PDSet$ be a collection of $n$ pseudo-disks. For two integers
   $0 < t \leq k$, a subset $X \subseteq \PDSet$ is a
   \emphi{$( t,k)$-tuple} if
   \begin{compactenumi*}
       \item $\cardin{X} \leq t$,
       \item $\exists \pnt \in \cap_{\pdisk \in X} \pdisk$, and
       \item $\depthY{\pnt}{\PDSet} \leq \depthk$.
   \end{compactenumi*}
   Let $\tuplesZ{t}{\depthk}{n}$ be the set of all $(\leq t,k)$-tuples
   of $\PDSet$.  Then
   $\cardin{\tuplesZ{t}{\depthk}{n}} = O(n t k^{t-1})$. %
\end{lemma}
\begin{proof}
    Let $\Sample \subseteq \PDSet$ be a random sample, where each
    pseudo-disk is independently placed into $\Sample$ with
    probability $1/k$. Consider a specific $(t,k)$-tuple $X$, with a
    witness point $\pnt$ of depth $\leq k$. Without loss of
    generality, by moving $\pnt$, one can assume $\pnt$ is a vertex of
    $\ArrX{\PDSet}$.

    Let $\event_{X}$ be the event that $\pnt$ is of depth exactly $t$
    in $\ArrX{\Sample}$, and $X \subseteq \Sample$. For $\event_{X}$
    to occur, all the objects of $X$ need to be sampled into
    $\Sample$, and each of the at most $k-t$ pseudo-disks containing
    $\pnt$ in its interior are not in $\Sample$. Therefore
    \begin{equation*}
        \Prob{\event_{X}}%
        \geq%
        \frac{\pth{1-1/\depthk}^{\depthY{\pnt}{\PDSet} - |X|}}{k^{|X|}}%
        \geq%
        \frac{\pth{1-1/\depthk}^k}{k^{t}}%
        \geq%
        \frac{1}{e^2 \depthk^t}.
    \end{equation*}
    Note, that a vertex of depth $\leq k$ in $\ArrX{\Sample}$
    corresponds to at most one such an event happening. We thus have,
    by linearity of expectations, that
    \begin{equation*}
        \frac{\cardin{\tuplesZ{t}{\depthk}{n}}}{e^2 k^t}%
        \leq %
        \Ex{\bigl.\cardin{\vDY{t}{\ArrX{\Sample}}}}%
        =%
        O(tn/k),
    \end{equation*}
    by \lemref{num:vertices:depth}.
\end{proof}

\begin{lemma}
    \lemlab{num:edges:pdisks}%
    Let $\Graph_i = (\PDSet_i, E_i)$ be the intersection graph of the
    pseudo-disks of $\PDSet_i$ (in the $i$\th iteration).  If
    $\ArrX{\PDSet_i}$ has maximum depth $\depthk$, then
    $\cardin{E_i} = O(n\depthk)$. Furthermore,
    $\indep(\Graph_i) = \Omega(n/\depthk)$, where $\indep(\Graph_i)$
    denotes the size of the largest independent set in $\Graph_i$.
\end{lemma}
\begin{proof}
    The first claim follows from \lemref{num:edges:depth}.
    Indeed, $\cardin{E_i} = \tuplesZ{2}{\depthk}{n} = O(n\depthk)$ ---
    since every intersecting pair of pseudo-disks induces a
    corresponding $(2,\depthk)$-tuple.

    For the second part, \Turan's Theorem states that any graph has an
    independent set of size at least $n/\pth{\avgdeg(\Graph_i) + 1}$,
    where $\avgdeg(\Graph_i) = 2\cardin{E_i}/n \leq c \depthk$ is the
    average degree of $\Graph_i$ and $c$ is some constant. It follows
    that $\indep(\Graph_i) \geq n/(c\depthk + 1) = \Omega(n/\depthk)$.
\end{proof}

The challenge in analyzing the greedy algorithm in 3D is that
mutual visibility between pairs of points is not necessarily lost as
the inner approximation grows.  As an alternative, consider the
\emph{hypergraph} $\Hgraph_i = (\PDSet_i, \hedges_i)$, where a triple
of pseudo-disks $\pdisk_1, \pdisk_2, \pdisk_3 \in \PDSet_i$ form a
hyperedge $\brc{\pdisk_1,\pdisk_2,\pdisk_3} \in \hedges_i$ $\iff$
$\pdisk_1 \cap \pdisk_2 \cap \pdisk_3 \neq \varnothing$ (this is
equivalent to the condition that the corresponding triple of points
span a triangle which does not intersect $\inapx_i$).

As in the analysis of the algorithm in 2D, we first bound the
number of edges in $\Hgraph_i$ and then argue that enough progress
is made in each iteration.

\begin{lemma}
    \lemlab{num:triples:pdisks}%
    Let $\Hgraph_i = (\PDSet_i, \hedges_i)$ be the hypergraph in
    iteration $i$, and let $\Graph_i$ be the corresponding
    intersection graph of $\PDSet_i$.  If $\ArrX{\PDSet_i}$ has
    maximum depth $\depthk$, then
    $\cardin{\hedges_i} = O(\indep(\Graph_i) \depthk^3)$.
\end{lemma}
\begin{proof}
    \lemref{num:edges:pdisks} implies that $\Graph_i$ has an
    independent set of size $\Omega(f_i/k)$, where
    $f_i = \cardin{\PDSet_i}$. \lemref{num:edges:depth} implies that
    $\cardin{\hedges_i} \leq \cardin{\tuplesZ{3}{\depthk}{f_i}} =
    O(f_i\depthk^2) = O(\indep(\Graph_i) \depthk^3)$.
\end{proof}

The following is a consequence of the Colorful \Caratheodory Theorem
\cite{b-gct-82}, see Theorem 9.1.1 in \cite{m-ldg-02}.

\begin{theorem}
    \thmlab{cpnt:many:simplices}%
    Let $\PS$ be a set of $n$ points in $\Re^d$ and $\cpnt$ be the
    centerpoint of $\PS$. Let $\SSet = \binom{\PS}{d+1}$ be the
    set of all $d+1$ simplices induced by $\PS$. Then for
    sufficiently large $n$, the number of simplices in $\SSet$ that
    contain $\cpnt$ in their interior is at least $c_d n^{d+1}$, where
    $c_d$ is a constant depending only on $d$.
\end{theorem}

Next, we argue that in each iteration of the greedy algorithm, a
constant fraction of the edges in $\Hgraph_i$ are removed. The
following is the higher dimensional version of
\lemref{segments:intersect}.

\begin{lemma}
    \lemlab{cpnt:many:simplex:faces}%
    Let $\PS$ be a set of $n$ points in $\Re^3$ lying above the
    $xy$-plane, $\cpnt$ be the centerpoint of $\PS$ and
    $\TSet = \binom{P}{3}$ be the set of all triangles induced by
    $\PS$. Next, consider any point $\pntB$ on the $xy$-plane.  Then
    the segment $\cpnt\pntB$ intersects at least $\Omega(n^3)$
    triangles of $\TSet$.
\end{lemma}
\begin{proof}
    Let $\SSet = \binom{\PS}{d+1}$ be the set of all simplices
    induced by $\PS$.  \thmref{cpnt:many:simplices} implies that the
    centerpoint $\cpnt$ is contained in $n^4/\constA$ simplices of
    $\SSet$ for some constant $\constA > 1$. Let $\simplex$ be a
    simplex that contains $\cpnt$ and observe the segment $\cpnt\pntB$
    must intersect at least one of the triangular faces $\tau$ of
    $\simplex$. As $\simplex \in \SSet$, charge this simplex
    $\simplex$ to the triangular face $\tau$.  Applying this counting
    to all the simplices containing $\cpnt$, implies that at least
    $n^4/\constA$ charges are made. On the other hand, a triangle
    $\tau$ can be charged at most $n-3$ times (because a simplex can
    be formed from $\tau$ and one other additional point of
    $\PS$). It follows that $\cpnt\pntB$ intersects at least
    $(n^4/\constA)/ (n-3) = \Omega(n^3)$ triangles of $\TSet$.
\end{proof}

\begin{lemma}
    \lemlab{depth:reduce:3d}%
    In each iteration of the greedy algorithm, the number of edges in
    the hypergraph $\Hgraph_i = (\PDSet_i, \hedges_i)$ decreases by at
    least $\Omega(\depthk^3)$, where $\depthk$ is the maximum depth of
    any point in $\ArrX{\PDSet_i}$.
\end{lemma}
\begin{proof}
    Recall that $\USet^+ = \USet_i \cap \LineB^+$ is the current set
    of unclassified points and $\LineB$ is the plane tangent to
    $\inapx_i$, where $\LineB^+$ is the closed halfspace that avoids
    the interior of $\inapx_i$ and contains the largest number of
    unlabeled points.  Note that $|\USet^+| \geq \depthk$.

    In a \RemoveOp{} operation, arguing as in \lemref{depth:reduce},
    implies that the number of points of $\USet^+$ that are discarded
    is at least $t_i \geq \depthk/4$. Since all of the discarded points are in
    a halfspace avoiding $\inapx_i$, it follows that all the triples
    they induce are in $\Hgraph_i$. Namely, at least
    $\binom{t_i}{3} = \Omega(k^3)$ hyperedges get discarded.

    In an \ExpandOp{} operation, the centerpoint $\cpnt$ of $\USet^+$
    is added to the current inner approximation $\inapx_i$.  Since all
    of the points of $\USet^+$ lie above the plane $\LineB$, applying
    \lemref{cpnt:many:simplex:faces} on $\USet^+$ with the centerpoint
    $\cpnt$ and a point lying on the plane $\LineB$ inside the
    (updated) inner approximation, we deduce that at least
    $\Omega(\depthk^3)$ hyperedges are removed.
\end{proof}

\begin{theorem}
    \thmlab{greedy-method-3d}%
    Let $\body \subseteq \Re^3$ be a convex body provided via a
    separation oracle, and let $\PS$ be a set of $n$ points in
    $\Re^3$.  The greedy classification algorithm performs
    ${O\bigl((\indexX{\PS}+1) \log n\bigr)}$ oracle queries. The
    algorithm correctly identifies all points in $\PS \cap \body$ and
    $\PS \setminus \body$.
\end{theorem}
\begin{proof}
    The proof is essentially the same as \thmref{greedy-method}.
    Arguing as in \lemref{iterations:empty:inapx} implies that there
    are at most $O(\log n)$ iterations (and thus also oracle queries)
    in which the inner approximation is empty.

    Now consider the hypergraph $\Hgraph_1 = (\PDSet_1, \hedges_1)$ at
    the start of the algorithm execution.  As the algorithm
    progresses, both vertices and hyperedges are removed from the
    hypergraph. Let $\Hgraph_i = (\PDSet_i, \hedges_i)$ denote the
    hypergraph in the $i$\th iteration of the algorithm.  Recall that
    $\PDSet_i$ is a set of pseudo-disks associated with each of the
    points yet to be classified. Observe that any independent set of
    pseudo-disks in the corresponding {intersection graph} $\Graph_i$
    corresponds to an independent set of points with respect to the
    inner approximation $\inapx_i$, and as such is a subset of points
    in convex position. Therefore, the size of any such independent
    set is bounded by $\indexX{\PS}$.

    Let $\depthk_i$ denote the maximum depth of any vertex in the
    arrangement $\ArrX{\PDSet_i}$.  \lemref{num:triples:pdisks}
    implies that
    $\cardin{\hedges_i} = O\pth{\indexX{\PS} \depthk_i^3}$.
    \lemref{depth:reduce:3d} implies that the number of hyperedges in
    the $i$\th iteration decreases by at least
    $\Omega(\depthk_i^3)$. Namely, after $O( \indexX{\PS})$
    iterations, the maximum depth is halved.  It follows that after
    $O( \indexX{\PS} \log n)$ iterations, the maximum depth is zero,
    which implies that all the points are classified. Since the
    algorithm performs one query per iteration, the claim follows.
\end{proof}

\section{An instance-optimal approximation in two dimensions}
\seclab{improved:2d}

Before discussing the improved algorithm, we present a lower bound on
the number of oracle queries performed by any algorithm that
classifies all the given points. We then present the improved
algorithm, which matches the lower bound up to a factor of $O(\log^2n)$.

\subsection{A lower bound}
\seclab{lower:bound}

Given a set $\PS$ of points in the plane, and a convex body $\body$,
the \emphi{outer fence} of $\PS$ is a closed convex polygon $\Fout$
with minimum number of vertices, such that $\body \subseteq \Fout$ and
$\body \cap \PS = \Fout \cap \PS$. Similarly, the \emphi{inner
   fence} is a closed convex polygon $\Fin$ with minimum number of
vertices, such that $\Fin \subseteq \body$ and
$\body \cap \PS = \Fin \cap \PS$. Intuitively, the outer fence
separates $\PS \setminus \body$ from $\partial \body$, while the
inner fence separates $\PS \cap \body$ from $\partial \body$.  The
\emphi{separation price} of $\PS$ and $\body$ is
\begin{equation*}
    \priceY{\PS}{\body} = \nVX{ \Fin} + \nVX{ \Fout},
\end{equation*}
where $\nVX{F}$ denotes the number of vertices of a polygon $F$.  See
\figref{easy:not:easy} for an example.

\begin{figure}[t]
    \hfill%
    \includegraphics[page=1,width=0.3\linewidth]{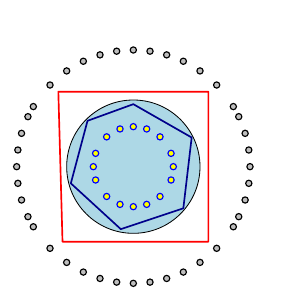}%
    \hfill%
    \includegraphics[page=2,width=0.3\linewidth]{figs/splitting}%
    \hfill%
    \includegraphics[page=3,,width=0.3\linewidth]{figs/splitting}
    \hfill%
    \phantom{}
    \caption{The separation price, for the same point set, is
       different depending on how ``tight'' the body is in relation to
       the inner and outer point set.}
    \figlab{easy:not:easy}
\end{figure}

\begin{lemma}
    \lemlab{lower:bound}
    Let $\body$ be a convex body provided via a separation oracle,
    and let $\PS$ be a point set in the plane. Any algorithm that
    classifies the points of $\PS$ in relation to $\body$, must perform
    at least $\priceY{\PS}{\body}$ separation oracle queries.
\end{lemma}
\begin{proof}
    Consider the set $Q$ of queries performed by the optimal algorithm
    (for this input), and split it, into the points inside and outside
    $\body$. The set of points inside, $\Qin = Q \cap \body$
    has the property that $\Qin \subseteq \body$, and furthermore
    $\CHX{\Qin} \cap \PS = \body \cap \PS$ --- otherwise, there
    would be a point of $\body \cap \PS$ that is not
    classified. Namely, the vertices of $\CHX{\Qin}$ are vertices of a
    fence that separates the points of $\PS$ inside $\body$ from the
    boundary of $\body$. As such, we have that
    $\cardin{\Qin} \geq \nVX{\CHX{\Qin}} \geq \nVX{\Fin}$.

    Similarly, each query in $\Qout = Q \setminus \Qin$ gives rise to
    a separating halfplane. The intersection of the corresponding
    halfplanes is a convex polygon $H$ that contains $\body$, and
    furthermore contains no point of $\PS \setminus \body$. Namely,
    the boundary of $H$ behaves like an outer fence. As such, we have
    $\cardin{\Qout} \geq \nVX{H} \geq \nVX{\Fout}$.

    Combining, we have that
    $\cardin{Q} = \cardin{\Qin} + \cardin{\Qout} \geq \nVX{\Fin} +
    \nVX{\Fout} = \priceY{\PS}{\body}$, as claimed.
\end{proof}

\begin{remarks}
\begin{compactenumi}
\item
Naturally the separation price, and thus the proof of the lower bound,
generalizes to higher dimensions. See \defref{lower:bound:high}
and \lemref{lower:bound:high}.

\item The lower bound only holds for $d \geq 2$. In 1D, the problem
can be solved using $O(\log n)$ queries with binary search. The above
would predict that any algorithm needs $\Omega(1)$ queries. However it
is not hard to argue a stronger lower bound of $\Omega(\log n)$.

\item
In \apndref{ex:sep:pr}, we show that when $\PS$ is a set of $n$ points
chosen uniformly at random from a square and $\body$ is a smooth convex
body, $\Ex{\priceY{\PS}{\body}} = O(n^{1/3})$. Thus, when the
points are randomly chosen, one can think of $\priceY{\PS}{\body}$ as
growing sublinearly in $n$.
\end{compactenumi}
\end{remarks}

\subsection{Useful operations}
\seclab{useful}

We start by presenting some basic operations that the new algorithm
will use.

\subsubsection{A directional climb}
\seclab{dir:climb}

Given a direction $v$, a \emphi{directional climb} is a sequence of
iterations, where in each iteration, the algorithm finds the extreme
line $\ell$ perpendicular to $v$, that is tangent to the inner approximation
$\inapx$. The algorithm then performs an iteration with $\ell$, as
described in \secref{round}, which we now recall. Specifically,
the algorithm computes the centerpoint $\query$ of all points in the halfspace
bounded by $\ell$ that avoids $\body$. Depending on whether $\query \in \body$,
we either perform $\ExpandOp{}(\query)$ or $\RemoveOp{}(\ell)$ (see
\secref{operations}). We then classify points accordingly and recompute $\ell$
with the updated inner approximation $\inapx$. See \figref{directional:climb}
for an illustration. The directional climb ends when the outer halfspace
induced by this line contains no unclassified point.

\begin{figure}[t]
    \centerline{%
       \includegraphics%
       {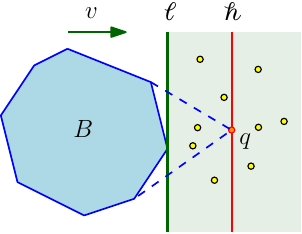}
    }%
    \caption{A directional climb. An iteration is done using the line
       $\Line$. After updating $\inapx$ to include the query $q$, the
       algorithm chooses a new extreme line $\LineA$ tangent to
       $\inapx$ in the direction of $v$.}
  \figlab{directional:climb}
\end{figure}

\begin{lemma}
    \lemlab{dir:climb}
    A directional climb requires $O( \log n)$ oracle queries.
\end{lemma}
\begin{proof}
    Consider the tangent to $\inapx$ in the direction of $v$. At each
    iteration, we claim the number of points in this halfplane is
    reduced by a factor of $1/3$. Indeed, if the query (i.e.,
    centerpoint) is outside $\body$ then at least a third of these
    points got classified as being outside. Alternatively, the tangent
    halfplanes moves in the direction of $v$, since the query point is
    inside $\body$. But then the new halfspace contains at most $2/3$
    fraction of the previous point set --- again, by the centerpoint
    property.
\end{proof}

\subsubsection{Line cleaning}

\begin{figure}[t]
    \hfill%
    \includegraphics[page=1]%
    {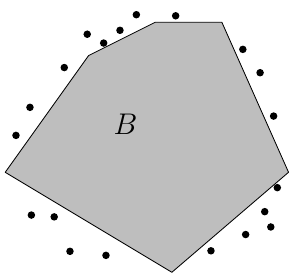}%
    \hfill%
    \includegraphics[page=2]%
    {figs/pocket}%
    \hfill\phantom{}%
    \caption{Unclassified points and their pockets.}
    \figlab{pockets}%
\end{figure}

A \emphi{pocket} is a connected region of
$\CHX{\USet \cup \inapx} \setminus \inapx$, see \figref{pockets}.  For
the set $\PS$ of input points, consider the set of all lines
\begin{equation}
    \LinesX{\PS}%
    =%
    \Set{\mathrm{line}(\pnt, \pntB)}{\pnt,\pntB \in \PS}%
    \eqlab{lines:x}%
\end{equation}
they span.

Let $\Line$ be a line that splits a pocket $\Pocket$ into two regions,
and furthermore, it intersects $\inapx$. Let
$\Interval = \Line \cap \Pocket$, and consider all the intersection
points of interest along $\Interval$ in this pocket. That is,
\begin{equation*}
    \IPSetZ{\Pocket}{\Line}{\PS}%
    =%
    \Interval \cap \LinesX{\PS}%
    =%
    \Set{\bigl. (\Pocket \cap \Line) \cap \LineA }{
       \LineA \in \LinesX{\PS}}.
\end{equation*}
In words, we take all the pairs of points of $\PS$ (each such pair
induces a line) and we compute the intersection points of these lines
with the interval $\Interval$ of interest.  Ordering the points of
this set along $\Line$, a prefix of them is in $\body$, while the
corresponding suffix are all outside $\body$. One can easily compute
this prefix/suffix by doing a binary search, using the separation
oracle for $\body$ --- see the lemma below for details. Each answer
received from the oracle is used to update the point set, using
\ExpandOp{} or \RemoveOp{} operations, as described in
\secref{operations}. We refer to this operation along $\Line$ as
\emphi{cleaning} the line $\Line$. See \figref{clean}.

\begin{figure}[h]
    \centerline{%
       \hfill%
       \begin{minipage}{0.3\linewidth}
           \includegraphics[page=1,width=0.99\linewidth]{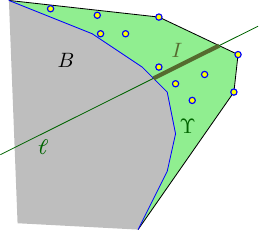}%
       \end{minipage}
       \hfill%
       \begin{minipage}{0.3\linewidth}
           \includegraphics[page=3,width=0.99\linewidth]{figs/ps_2}%
       \end{minipage}
       \hfill%
       \begin{minipage}{0.3\linewidth}
           \includegraphics[page=4,width=0.92\linewidth]{figs/ps_2}\\%
           \includegraphics[page=5,width=0.92\linewidth]{figs/ps_2}%
       \end{minipage}
       \hfill\phantom{}%
    } \vspace*{-0.25cm}%
    \caption{Line cleaning. All the intersection points of interest
       along $\Line$ are classified. The binary search results in the
       oracle returning a line $\LineA$ that separates the points
       outside from the points inside. }
    \figlab{clean}%
\end{figure}

\begin{lemma}
    Given a pocket $\Pocket$, and a splitting line $\Line$, one can
    clean the line $\Line$ --- that is, classify all the points of
    $\IPSet = \IPSetZ{\Pocket}{\Line}{\PS}$ using
    $O\pth{ \log n \bigr.}$ oracle queries.  By the end of this
    process, $\Pocket$ is replaced by two pockets, $\Pocket_1$ and
    $\Pocket_2$ that do not intersect $\Line$. The pockets $\Pocket_1$
    or $\Pocket_2$ may be empty sets.
\end{lemma}

\begin{proof}
    First, we describe the line cleaning procedure in more detail.
    The algorithm maintains, in the beginning of the $i$\th iteration,
    an interval $\IntervalA_i$ on the line $\Line$ containing all the
    points of $\IPSet$ that are not classified yet. Initially,
    $\IntervalA_1 = \Pocket \cap \Line$. One endpoint, say
    $\pnt_i \in \IntervalA_i$ is on $\partial \inapx_i$, and the
    other, say $\pnt_i'$, is outside $\body$, where $\inapx_i$ is the
    inner approximation in the beginning of the $i$\th iteration.

    In the $i$\th iteration, the algorithm computes the set
    $\IPSet_i = \IntervalA_i \cap \IPSet$. If this set is empty, then
    the algorithm is done. Otherwise, it picks the median point
    $\pntC_i$, in the order along $\Line$ in $\IPSet_i$, and queries
    the oracle with $\pntC_i$. There are two possibilities: %
    \medskip%
    \begin{compactenumA}
        \item If $\pntC_i \in \body$ then the algorithm sets
        $\IPSet_{i+1} = \IPSet_i \setminus [\pnt_i, \pntC_i)$, and
        $\IntervalA_{i+1} = \IntervalA_i \setminus [\pnt_i,\pntC_i)$.

        \smallskip%
        \item If $\pntC_i \notin \body$, then the oracle provided a
        closed halfspace $h^+$ that contains $\body$. Let $h^-$ be the
        complement open halfspace that contains $\pntC_i$. The
        algorithm sets $\IPSet_{i+1} = \IPSet_{i} \setminus h^-$ and
        $\IntervalA_{i+1} = \IntervalA_i \cap h^+$.
    \end{compactenumA}
    \medskip%
    This resolves the status of at least half the points in
    $\IPSet_i$, and shrinks the active interval. The algorithm repeats
    this till $\IPSet_i$ becomes empty.  Since
    $\cardin{\IPSet} = O(n^2)$, this readily implies that the
    algorithm performs $O( \log n)$ iterations.

    We now argue that the pocket is split --- that is, $\Pocket_1$ and
    $\Pocket_2$ do not intersect $\Line$. Assume that it is false, and
    let $\inapx'$ be the inner approximation after this procedure is
    done. Let $L$ (resp.~$R$) be the points of
    $\USet_\Pocket = \USet \cap \Pocket$
    that are unclassified on one side (resp.~other side) of
    $\Line$.  If the pocket is not split, then there are two points
    $\pnt \in L$ and $\pntB \in R$, such that
    $\pnt\pntB \cap \inapx' = \emptyset$, and
    $\partial \CHX{\inapx' \cup L \cup R}$ intersects $\Line$ at the
    point $\pntC = \pnt \pntB \cap \Line$.  However, by construction,
    the point $\pntC \in \IPSet$. As such, the point $\pntC$ is now
    classified as either being inside or outside $\body$, as it is a
    point in $\IPSet$. If $\pntC$ is outside, then the halfplane $h^-$
    that classified it as such, must had classified either $\pnt$ or
    $\pntB$ as being outside $\body$, which is a contradiction. The
    other option, is that $\pntC$ is classified as being inside, but
    then, it is in $\inapx'$, which is again a contradiction, as it
    implies that $\inapx'$ intersects the segment $\pnt \pntB$.
\end{proof}

\subsubsection{Vertical pocket splitting}
\seclab{pocket:split}%

Consider a pocket $\Pocket$ such that all of its points lie
vertically above $\inapx$, and the bottom of $\Pocket$ is part of a
segment of $\partial \inapx$, see \figref{v:pocket}. Such a pocket can
be viewed as being defined by an interval on the $x$-axis
corresponding to its two vertical walls.  Let $\USet_\Pocket$ be the
set of unclassified points in this pocket. In each iteration, the
algorithm computes the centerpoint $\query$ of $\USet_\Pocket$, and queries
the separation oracle for the label of $\query$. As long as the query
point is outside $\body$, the algorithm performs a \RemoveOp{} operation
using the returned separating line.

When the oracle returns that the query point $\query$ is inside
$\body$, the algorithm computes the vertical line $\Line_\query$
through $\query$.  The algorithm now performs line cleaning on this
vertical line.  This operation splits $\Pocket$ into two sub-pockets.
Crucially, since $\query$ was a centerpoint for $\USet_\Pocket$,
the number of points in each of the two sub-pockets is at most
$2\cardin{\USet_\Pocket}/3$.  See \figref{v:pocket}.

\begin{figure}[t]
    \phantom{}\hfill%
    \includegraphics[page=1,width=0.4\linewidth]%
    {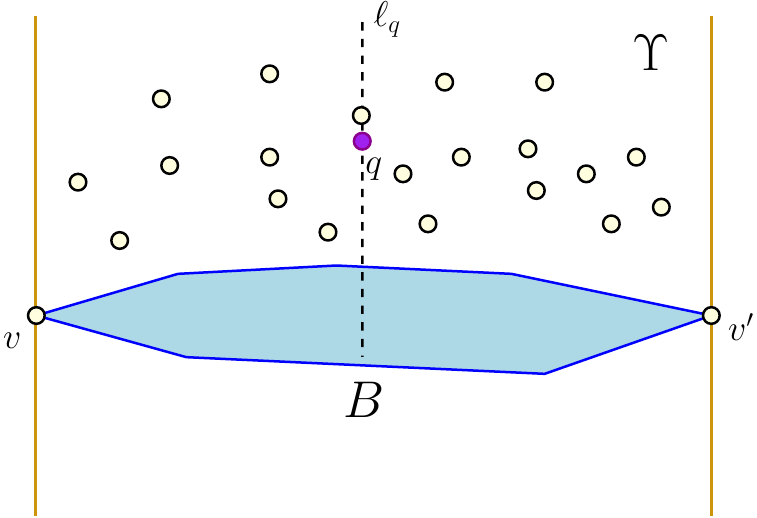}%
    \hfill%
    \includegraphics[page=2,width=0.4\linewidth]%
    {figs/optimal-2d}%
    \hfill
    \caption{Vertical pocket splitting. In this example, the
       centerpoint $\query$ lies inside $\body$. Thus we construct the
       vertical line $\ell_\query$ through $\query$ (left). Next, we
       perform a line cleaning operation on $\ell_\query$. This splits
       the original pocket $\Pocket$ into two new pockets $\Pocket_1$,
       $\Pocket_2$, while classifying some points in the process
       (right). Observe that the unclassified points in $\Pocket_1$
       and $\Pocket_2$ are no longer mutually visible to each other
       after the line cleaning operation.}%

    \figlab{v:pocket}%
\end{figure}

\subsection{The algorithm}
\seclab{improved:alg}

The algorithm starts in the same way as the greedy algorithm of
\secref{round}, which we restate for convenience.
Recall that $\USet$ is the set of unclassified points (initially
$\USet = \PS$). At all times, the algorithm maintains
the inner approximation $\inapx \subseteq \body$.
At the beginning, $\inapx$ is uninitialized. The algorithm computes
the centerpoint $\query$ of $\USet$ and queries the
oracle for the label of $\query$. While $\query$ is outside,
we classify the appropriate set of points as outside (according to
the separating hyperplane returned from the oracle), update
$\USet$, and repeat. As soon as the computed centerpoint $\query$
lies in $\body$, we set $\inapx = \query$ and continue to the next stage.

Next, the algorithm performs two directional climbs (\secref{dir:climb})
in the positive and negative directions of the $x$-axis. This uses
$O( \log n)$ oracle queries by \lemref{dir:climb} and results in a
computed segment $\pntD \pntD' \subseteq \body$, where $\pntD, \pntD'$
are vertices of the inner approximation $\inapx$, such that all
unclassified points lie in the strip induced by the vertical line through
$\pntD$ and the vertical line through $\pntD'$, see also \figref{v:pocket}.

The algorithm now handles all points of $\USet$ lying above
$\pntD \pntD'$ (the points below the line are handled in a
similar fashion). Let $\inapx^+$ be the set of vertices of $\inapx$ in
the top chain. Note that $\inapx^+$ consists of at most $O(\log n)$
vertices.  For each vertex $v$ of $\inapx^+$, the algorithm performs
line cleaning on the vertical line going through $v$.  This results in
$O(\log n)$ vertical pockets, where all vertical lines passing
originally through $\inapx^+$ are now clean.

The algorithm repeatedly picks a vertical pocket. If the pocket
contains less than three points the algorithm queries the oracle for the
classification of these points, and continues to the next pocket.
Otherwise, the algorithm performs a vertical pocket splitting
operation, as described in \secref{pocket:split}. The algorithm stops
when there are no longer any pockets (i.e., all the points above the
segment $\pntD \pntD'$ are classified). The algorithm then runs the
symmetric procedure below this segment $\pntD \pntD'$.

\subsection{Analysis}

\begin{figure}[t]
    \centerline{%
       \hfill%
       \begin{minipage}{0.3\linewidth}
           \centering%
           \includegraphics[page=1,width=0.99\linewidth]%
           {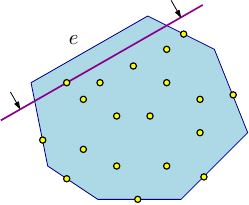}%
       \end{minipage}
       \hfill%
       \begin{minipage}{0.4\linewidth}
          \centering
          \includegraphics[page=2,width=0.99\linewidth]%
          {figs/inner_fence}\\%
      \end{minipage}
       \hfill%
       \phantom{}
    }
    \caption{Constructing the polygon $\PolygonA$ from an inner
    fence $\Polygon$.}%
    \figlab{inner:fence}%
\end{figure}

\begin{lemma}
    \lemlab{separate:but}%
    Given a point set $\PS$, and a convex polygon $\Polygon$ that is
    an inner fence for $\PS \cap \body$; that is,
    $\PS \cap \body \subseteq \Polygon \subseteq \body$. Then, there
    is a convex polygon $\PolygonA$, such that
    \begin{compactenumA}
        \item
        $\PS \cap \body \subseteq \PolygonA \subseteq \Polygon$.
        \item $\nVX{\PolygonA} \leq 2\nVX{\Polygon}$ (where $\nVX{Q}$
        denotes the number of vertices of the polygon $Q$).
        \item Every edge of $\PolygonA$ lies on a line of
        $\LinesX{\PS}$, see \Eqref{lines:x}.
    \end{compactenumA}
\end{lemma}
\begin{proof}
    Any edge $\edge$ of $\Polygon$ that does not contain any point of
    $\PS$ on it can be moved parallel to itself into the polygon
    until it passes through a point of $\PS$. Next, split the edges
    that contain only a single point of $\PS$, by adding this point
    as a vertex.

    Consider a vertex $v$ of the polygon that is not in $\PS$ ---
    and consider the two adjacent vertices $u,w$, which must be in
    $\PS$. If $\triangle uvw \setminus uw$ contains no point of $\PS$,
    then we delete $v$ from the polygon and replace it by the edge
    $uw$. Otherwise, move $v$ towards $u$, until the edge $vw$ hits a
    point of $\PS$. Next, move $v$ towards $w$, till the edge $vu$
    hits a point of $\PS$. See \figref{inner:fence}.

    Repeating this process so that all edges contain two points of
    $\PS$ means that properties (A) and (C) are met.
    Additionally, the number of edges of the new polygon $\PolygonA$
    is at most twice the number of edges of $\Polygon$,
    implying property (B).
\end{proof}

Consider the inner and outer fences $\Fin$ and
$\Fout$ of $\PS$ in relation to $\body$. Applying
\lemref{separate:but} to $\Fin$, results in a convex polygon
$\PolygonA$ that separates $\PS \cap \body$ from
$\partial \body$, that has at most $2 \nVX{\Fin}$ vertices. Let
$\VSet$ be the set of all vertices of the polygons $\Fin, \Fout$
and $\PolygonA$.

The following two Lemmas state that if a vertical pocket $\Pocket$
containing no vertex of $\VSet$, then all points in $\Pocket$ can be
classified using $O(\log n)$ oracle queries. Finally, we analyze the
scenario when $\Pocket$ contains at least one vertex of $\VSet$.

\begin{lemma}
  \lemlab{no:vertex:outside}
  Let $\Pocket$ be a vertical pocket created during the algorithm with
  current inner approximation $\inapx$. Suppose that
  $\VSet \cap \Pocket = \varnothing$, then all points in
  $\PS \cap \Pocket$ are outside $\body$.
\end{lemma}
\begin{proof}
  Assume without loss of generality that $\Pocket$ lies above $\inapx$.
  Let $\USet = \PS \cap \Pocket$ be the set of unclassified points in
  the pocket. Note that $\Pocket$ is bounded by two vertical lines that
  were previously cleaned.

  By assumption, $\Pocket$ does not contain any vertex of $\PolygonA$.
  It follows that there is a single edge of $\PolygonA$ that
  intersects the two vertical lines bounding $\Pocket$. Let
  $\pntC_L, \pntC_R$ be these two intersection points, one lying on
  each line. By definition, we have $\pntC_L, \pntC_R \in \body$.
  Furthermore, $\pntC_L, \pntC_R$ lie on lines of $\LinesX{\PS}$ by
  construction of $\PolygonA$. Since both vertical lines bounding
  $\Pocket$ were cleaned, it must be that the segment
  $\pntC_L \pntC_R \subseteq \inapx$. Since all points of $\USet$ are
  above $\inapx$, this implies that $\USet$ lies above
  $\pntC_L \pntC_R$ and thus above $\PolygonA$. Namely, all points of
  $\USet$ are outside $\body$.
\end{proof}

\begin{lemma}
  \lemlab{no:vertex:classify} Let $\Pocket$ be a vertical pocket with
  $\VSet \cap \Pocket = \varnothing$. Then during the vertical pocket
  splitting operation of \secref{pocket:split} applied to $\Pocket$,
  all oracle queries are outside $\body$. In particular, all points of
  $\PS \cap \Pocket$ are classified after $O(\log n)$ oracle queries.
\end{lemma}
\begin{proof}
  Let $\USet = \PS \cap \Pocket$. By \lemref{no:vertex:outside}, all
  points of $\USet$ lie outside $\body$. Assume that the first
  statement of the Lemma is false, and let $\USet' \subseteq \USet$ be
  the set of unclassified points such that $\query$ was the centerpoint
  for $\USet'$ and $\query \in \body$. Now $\query$ is inside a
  triangle induced by three points of $\USet'$. Namely, there are (at
  least) two points outside $\body$ in this pocket that are not
  mutually visible to each other with respect to $\body$. But this implies
  that $\Fout$ must have a vertex somewhere inside the vertical pocket
  $\Pocket$, which is a contradiction.

  Hence, all oracle queries made by the algorithm are outside $\body$.
  Each such query results in a constant reduction in the size of
  $\USet$, since the query point is a centerpoint of the unclassified
  points. It follows that after $O(\log\cardin{\USet}) = O(\log n)$
  queries, all points in $\Pocket$ are classified.
\end{proof}

\begin{theorem}
  \thmlab{improved:alg:2d}
  Let $\body$ be a convex body provided via a separation oracle, and
  let $\PS$ be a set of $n$ points in the plane. The improved
  classification algorithm performs
  \begin{math}
      O\pth{ \bigl[ 1 +\priceY{\PS}{\body}\bigr] \log^2 n}
  \end{math}
  oracle queries. The algorithm correctly identifies all points
  in $\PS \cap \body$ and $\PS \setminus \body$.
\end{theorem}
\begin{proof}
    The initial stage involves two directional climbs and $O( \log n)$
    line cleaning operations, and thus requires $O( \log^2 n)$
    queries.

    A vertical pocket that contains a vertex of $\VSet$ is charged
    arbitrarily to any such vertex. Since the number of points in a
    pocket reduces by at least a factor of $1/3$ during a split
    operation, this means that a vertex of $\VSet$ is charged at most
    $O(\log n)$ times. Each time a vertex gets charged, it has to pay
    for the $O(\log n)$ oracle queries that were issued in the
    process of creating this pocket, and later on for the price of
    splitting it. Thus, we only have to account for queries performed
    in vertical pockets that do not contain a vertex of $\VSet$.
    By \lemref{no:vertex:classify}, such a pocket will have all
    points inside it classified after $O(\log n)$ oracle queries.

    However, the above implies that there are at most
    $O([1+\priceY{\PS}{\body}] \log n)$ vertical pockets with no
    vertex of $\VSet$ throughout the algorithm execution.
    Since handling such a pocket requires $O( \log n)$ queries, the
    bound follows.
\end{proof}

\section{On emptiness variants in two dimensions}
\seclab{emptiness:2d}

Here, we present two instance-optimal approximation algorithms for
solving the following two variants:
\begin{compactenumA}
    \smallskip%
    \item Emptiness: Find a point $\pnt \in \PS \cap \body$, or
    using as few queries as possible, verify that
    $\PS \cap \body = \varnothing$.

    \smallskip%
    \item Reverse emptiness: Find a point
    $\pnt \in \PS \setminus (\PS \cap \body)$, or using as few
    queries as possible, verify that $\PS \cap \body = \PS$.
\end{compactenumA}
\smallskip%
For both variants we present $O( \log n )$ approximation
(the algorithm for emptiness is randomized), improving
over the general approximation algorithm of \secref{improved:2d}
which provides a $O( \log^2 n)$ approximation.

\subsection{Emptiness: Are all the points outside?}

Here we consider the problem of verifying that all the given points
are outside the convex body.

\myparagraph{Algorithm.}
The algorithm is a slight modification of the algorithm of
\secref{round}. Recall the two operations \ExpandOp{} and \RemoveOp{}
that the algorithm will need (\secref{operations}).

Initially, let $\USet = \PS$ be the set of unclassified points.
At every round, if the inner approximation $\inapx$ is empty, then the
algorithm sets $\USet^+ = \USet$. Otherwise, the algorithm picks a
line $\Line$ that is tangent to $\inapx$ with the largest number of
points of $\USet$ on the other side of $\Line$ than $\inapx$. Let
$\Line^-$ and $\Line^+$ be the two closed halfspace bounded by
$\Line$, where $\inapx \subseteq \Line^-$. The algorithm computes the
point set $\USet^+ = \USet \cap \Line^+$.  We have two cases:
\medskip%
\begin{compactenumA}[label=\Alph*.]
    \item Suppose $\cardin{\USet^+}$ is of constant size. The
    algorithm queries the oracle for the status of each of these
    points. If there exists a point $\pnt \in \USet^+$ which lies in $\body$,
    then we return $\pnt$ as the witness. Otherwise, for each
    $\pnt \in \USet^+$ we receive a separating line $\Line_\pnt$ from the
    oracle, and the algorithm executes $\RemoveOp{}(\Line_\pnt)$.

    \medskip%
    \item Otherwise, $\cardin{\USet+}$ does not have constant size.
    The algorithm chooses a point $\query \in \USet^+$ at random
    and queries the oracle using $\query$. If $\query \in \body$,
    we return $\query$ as the witness. Otherwise, we perform a
    \RemoveOp{} operation on the separating line returned.

    Next, we compute the centerpoint $\query$ of $\USet^+$ and
    query the oracle for the label of $\query$. Depending on the
    label of $\query$, the algorithm either executes
    $\ExpandOp{}(\query)$ or $\RemoveOp{}(\ell)$, where $\ell$
    is the separating line in the instance that $\query \not\in \body$.
\end{compactenumA}

\myparagraph{Analysis.}
Let $\Graph_i$ be the intersection graph (see \defref{visi:graph})
over the points outside $\body$ in the beginning of the $i$\th
iteration. We need the following technical Lemma.

\begin{lemma}
    \lemlab{indep:fin}%
   Suppose $\PS \cap \body = \varnothing$. Then at any iteration
   $i$, the largest independent set in the visibility graph $\Graph_i$
   is at most $\cardin{\Fout}$.
\end{lemma}
\begin{proof}
    For the body $\body$ and point set $\PS$, define the set
    $\RSet \subseteq \PS$ to be the maximum set of points such that
    no two points in $\RSet$ are visible with respect to $\body$.
    Observe that $\RSet$ corresponds to the maximum independent set
    in the visibility graph for $\PS$ with respect to the body
    $\body$. We claim $\cardin{\RSet} \leq \cardin{\Fout}$.  Suppose
    that $\cardin{\RSet} > \cardin{\Fout}$. Given the polygon
    $\Fout$, for each edge $e$ of $\Fout$ consider the line $\Line_e$
    through $e$ and let $\LineA_e^+$ be the halfspace bounded by
    $\Line_e$ that does not contain $\body$ in its interior. Then
    $\Set{\LineA_e^+}{e \in \Fout}$ covers the space
    $\Re^2 \setminus \intX{\body}$. By the hypothesis, one halfspace
    $\LineA_e^+$ must contain at least two points of $\RSet$. But
    then these two such points are visible with respect to $\body$,
    contradicting the definition of $\RSet$.

    We know that the size of the largest independent set (with respect
    to the current inner approximation $\inapx_i$) is monotone
    increasing over the iterations. Hence each independent set can be
    of size at most $\cardin{R} \leq \cardin{\Fout}$.
\end{proof}

\begin{lemma}
    \lemlab{greedy-method-empty}%
    Let $\body$ be a convex body provided via a separation oracle, and
    let $\PS$ be a set of $n$ points in the plane. The randomized greedy
    classification algorithm for emptiness performs
    ${O\bigl((\cardin{\Fout}+1) \log n\bigr)}$ oracle queries
    with high probability. The algorithm always correctly
    verifies that $\PS \cap \body = \varnothing$ or
    finds a witness point of $\PS$ inside $\body$.
\end{lemma}
\begin{proof}
    Suppose $\PS \cap \body = \varnothing$. Then \lemref{indep:fin}
    along with the proof of \thmref{greedy-method} implies the result,
    by replacing the quantity $\indexX{\PS}$ with $\cardin{\Fout}$.
    If $\PS \cap \body \neq \varnothing$, let $\USet^+$ be a set of
    points in the current iteration,
    $\USetin{+} = \USet^+ \cap \body$, and
    $\USetout{+} = \USet^+ \setminus \USetin{+}$. Observe that
    $\USetin{+}$ remains the same throughout the algorithm execution,
    while $\USetout{+}$ shrinks.  If
    $\cardin{\USetout{+}} > \cardin{\USet^+}/2$, then by
    \lemref{depth:reduce} the number of edges removed from $\Graph_i$
    is $\Omega\pth{\cardin{\USetout{+}}^2}$ (though the hidden
    constants will be smaller). Thus, after at most
    $O\bigl((\cardin{\Fout}+1)\log n \bigr)$ iterations, we must
    encounter an iteration in which there is a set of points
    $\USet^+$ with $\cardin{\USetout{+}} < \cardin{\USet^+}/2$.  Now
    the probability that our randomly sampled point lies in
    $\USetin{+}$ is at least 1/2. In particular, after an additional
    $O(\log n)$ iterations, the probability that we fail to find a
    witness point is at most $1/n^{\Omega(1)}$, thus implying the bound
    on the number of queries.
\end{proof}

\subsection{Reverse emptiness: Are all the points inside?}

Here we consider the problem of verifying that all the given points
are inside the convex body.

\subsubsection{Algorithm}

\myparagraph{Initialization.}
Let $\ch = \CHX{\PS}$. Define $\pntD, \pntD' \in \PS$ to be the
extreme left and right vertices of $\ch$. For the sake of exposition,
by a rotation of the space, we assume without loss of generality
that the segment $\pntD \pntD'$ is parallel to the $x$-axis.
Let $\pntD_1$ and $\pntD_2$ be the vertices adjacent to $\pntD$ on $\ch$.
Similarly define $\pntD'_1$ and $\pntD'_2$ for $\pntD'$. The algorithm
asks the oracle for the status of $\pntD$, $\pntD_1$, $\pntD_2$,
$\pntD'$, $\pntD'_1$, and $\pntD'_2$. If any of them are outside, the
algorithm halts and reports the witness found. Otherwise, all points must
lie either above or below the horizontal segment $\pntD\pntD'$. We now
describe how to handle the points above $\pntD\pntD'$ (the below case is
handled similarly).

Let $\ChUp$ be the polygonal chain that is the portion of $\ch$
contained inside the region bounded by the segment $vv'$ and the two
vertical lines passing through $\pntD$ and $\pntD'$.
Label the edges along $\ch^+$ by $\edgeB_1, \ldots, \edgeB_k$
clockwise from $\pntD$ to $\pntD'$. For $1 \leq i < j \leq k$, let
$\ChRangeY{i}{j}$ be the polygonal chain consisting of the consecutive
edges $\edgeB_i, \ldots, \edgeB_j$.  The algorithm now invokes the
following recursive procedure.

\myparagraph{Recursive procedure.}
A recursive call is described by two indices $(i,j)$, the goal is
to verify that all the points of $\PS$ lying below $\ChRangeY{i}{j}$
are inside $\body$.

For a given recursive instance $(i,j)$, the algorithm proceeds as
follows. Begin by computing the lines $\Line_i$ and $\Line_j$ through
the edges $\edgeB_i$ and $\edgeB_j$ respectively. Let
$\query = \Line_i \cap \Line_j$ be the point of intersection. The
algorithm asks the oracle for the status of $\query$. If $\query$ is
inside, then all points below $\ChRangeY{i}{j}$ must also be in
$\body$. The algorithm classifies the appropriate points and returns.
Otherwise $\query$ is outside, and generates two recursive calls. Let
$\ell = \floor{(i + j)/2}$ and $\edgeB_\ell = (x,y)$ be the middle edge
in the chain $\ChRangeY{i}{j}$. The algorithm queries the oracle with
$x$ and $y$. If either $x$ or $y$ is outside, the algorithm returns
the appropriate witness found. Otherwise $x$ and $y$ are both
inside. The algorithm recurses on the instances $(i, \ell)$ and
$(\ell, j)$.

\subsubsection{Analysis}

The analysis will use the polygon $\PolygonA$, as defined in
\lemref{separate:but}, applied to $\Fin$. Specifically, it is an inner
fence where $\cardin{\PolygonA} = O(\cardin{\Fin})$ and every edge of
$\PolygonA$ lies on a line of $\LinesX{\PS}$, see \Eqref{lines:x}. Note
that $\ch \subseteq \PolygonA$ and every edge of $\ch$ lies on a line
of $\LinesX{\PS}$. For each edge $\edge$ of $\PolygonA$, let
$\Line_\edge \in \LinesX{\PS}$ be the line containing $\edge$. We can
match every edge $\edge$ of $\PolygonA$ with the edge $\edgeB(\edge)$
of $\ch$ that lies on $\Line_\edge$. If an edge $\edgeB$ of $\ch$ is
matched to some edge of $\PolygonA$, we say that $\edgeB$ is
\emphi{active}. A recursive call $(i,j)$ is \emphi{alive} if the query
$\query = \Line_i \cap \Line_j$ generated is outside $\body$.

\begin{lemma}
    \lemlab{num:recursive:calls}%
    The number of alive recursive calls at the same recursive depth
    is at most $\cardin{\PolygonA} = O(\cardin{\Fin})$.
\end{lemma}
\begin{proof}
    Fix an alive recursive call $(i,j)$ with edges
    $\edgeB_i, \ldots, \edgeB_j$ of $\ch$. Suppose that none of these
    edges are active. Because $\PolygonA$ is an inner fence for $\PS$
    and $\body$, there must be a vertex $\pntD$ of $\PolygonA$ lying
    on or above the chain $\ChRangeY{i}{j}$. Let $\edge_1$ and
    $\edge_2$ be the edges adjacent to $v$ in $\PolygonA$. For
    $\ell = 1, 2$, consider $\edgeB(\edge_\ell)$, the edge of $\ch$
    matched to $e_\ell$. Since there are no active edges in
    $\ChRangeY{i}{j}$, we have
    $\edgeB(\edge_\ell) \not\in \{\edgeB_i, \ldots, \edgeB_j\}$ for
    $\ell = 1, 2$.  This readily implies that all vertices in the
    polygonal chain $\ChRangeY{i}{j}$ are contained in the wedge
    formed by $\pntD$ and the two edges $\edge_1$ and $\edge_2$. See
    \figref{edges}.

    \begin{figure}[h]
        \centering
        \includegraphics%
        {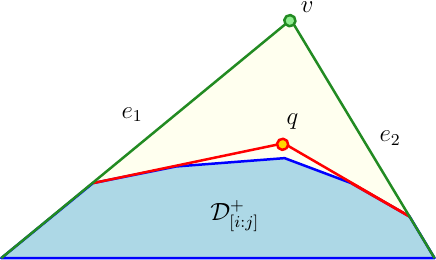}
        \caption{}
        \figlab{edges}
    \end{figure}

    In particular, the query $\query$ generated is inside $\PolygonA$
    and thus $\body$. Contradicting that the recursive call was alive.
    It follows that each alive recursive call must contain at least
    one active edge. The number of active edges is bounded by
    $\cardin{\PolygonA}$, implying the result.
\end{proof}

\begin{lemma}
    \lemlab{reverse:emptiness}%
    Let $\body$ be a convex body provided via a separation oracle, and
    let $\PS$ be a set of $n$ points in the plane. The
    classification algorithm for reverse emptiness performs
    $O\bigl(\cardin{\Fin} \log n\bigr)$ oracle queries.
    The algorithm correctly verifies that $\PS \cap \body = \PS$
    or finds a witness point of $\PS$ outside $\body$.
\end{lemma}
\begin{proof}
    Suppose all points of $\PS$ are inside $\body$. By
    \lemref{num:recursive:calls}, there are at most $O(\cardin{\Fin})$
    alive recursive calls at each level of the recursion tree. Since
    the depth of the recursion tree is $O(\log n)$, the number of
    total alive recursive calls throughout the algorithm is
    $O(\cardin{\Fin} \log n)$.  At each alive recursive call of the
    above algorithm, $O(1)$ queries are made. This implies the result.

    Otherwise not all points of $\PS$ are inside $\body$. At least
    one such point outside of $\body$ must be a vertex on the convex
    hull $\ch$. Hence after at most $O(\cardin{\Fin} \log n)$ oracle
    queries, this vertex will be queried and found to be outside
    $\body$.
\end{proof}

\section{Application: Minimizing a convex function}
\seclab{applications}

Suppose we are given a set of $n$ points $\PS$ in the plane and a convex
function $\fn : \Re^2 \to \Re$. Our goal is to compute the point in $\PS$
minimizing $\min_{\pnt \in \PS} \fn(\pnt)$. Given a point
$\pnt \in \Re^2$, assuming that we can evaluate $\fn$ and the derivative
of $\fn$ at $\pnt$ efficiently, we show that the point in $\PS$ minimizing
$\fn$ can be computed using $O(\indexX{\PS} \log^2 n)$ evaluations
to $\fn$ or its derivative.

\begin{definition}
    Let $\fn : \Re^d \to \Re$ be a convex function. For a number
    $c \in \Re$, define the \emphi{level set of $\fn$} as
    $\LvSetY{\fn}{c} = \Set{\pnt \in \Re^d}{\fn(\pnt) \leq c}$.  If
    $\fn$ is a convex function, then $\LvSetY{\fn}{c}$ is a convex set
    for all $c \in \Re$.
\end{definition}
\begin{definition}
    Let $\fn : \Re^d \to \Re$ be a convex (and possibly
    non-differentiable) function. For a point $\pnt \in \Re^d$, a
    vector $v \in \Re^d$ is a \emphi{subgradient} of $\fn$ at $\pnt$
    if for all $\pntq \in \Re^d$,
    $\fn(\pntq) \geq \fn(\pnt) + \DotProdY{v}{\pntq - \pnt}$.  The
    \emphi{subdifferential} of $\fn$ at $\pnt \in \Re^d$, denoted by
    $\partial \fn(\pnt)$, is the set of all subgradients $v \in \Re^d$
    of $\fn$ at $\pnt$.
\end{definition}

It is well known that when the domain for $\fn$ is $\Re^d$ and $\fn$ is
a convex function, then $\partial \fn(\pnt)$ is a non-empty set of all
$\pnt \in \Re^d$ (for example, see \cite[Chapter 3]{f-ina-13}).

Let $\opt = \min_{\pa \in \PS} \fn(\pa)$. We have that
$\LvSetY{\fn}{\opt} \cap \PS = \Set{\pa \in \PS}{\fn(\pa) = \opt}$
and $\LvSetY{\fn}{\opt'} \cap \PS = \varnothing$ for all $\opt' < \opt$.
Hence, the problem is reduced to determining the smallest value $r$
such that $\LvSetY{\fn}{r} \cap \PS$ is non-empty.

\begin{lemma}
  \lemlab{decision}
  Let $\PS$ be a collection of $n$ points in the plane.
  For a given value $r$, let $\body_r = \LvSetY{\fn}{r}$.
  The set $\body_r \cap \PS$ can be computed using
  $O(\indexX{\PS} \log n)$ evaluations to $\fn$ or
  its derivative. If $\TX$ is the time needed to evaluate $\fn$
  or its derivative, the algorithm can be implemented in
  $O(n\log^2 n\log\log n + \TX\cdot \indexX{\PS} \log n)$ expected
  time.
\end{lemma}
\begin{proof}
The Lemma follows by applying \thmref{greedy-method}.  Indeed, let
$\body_r = \LvSetY{\fn}{r}$ be the convex body of interest. It remains
to design a separation oracle for $\body_r$.

Given a query point $\query \in \Re^2$, first compute
$c = \fn(\query)$. If $c \leq r$, then report that
$\query \in \body_r$.  Otherwise, $c > r$. In this case, compute
some gradient vector $v$ in $\partial \fn(\query)$. Using the vector
$v$, we can obtain a line $\Line$ tangent to the boundary of
$\LvSetY{\fn}{c}$ at $\query$. As
$\LvSetY{\fn}{r} \subseteq \LvSetY{\fn}{c}$, $\Line$ is a separating
line for $\query$ and $\body_r$, as desired.
As such, the number of separation oracle queries needed to determine
$\body_r \cap \PS$ is bounded by $O(\indexX{\PS} \log n)$
by \thmref{greedy-method}.

The implementation details of \thmref{greedy-method} are given in
\lemref{impl-greedy}.
\end{proof}

\myparagraph{The algorithm.}
Let $\alpha = \min_{\pa \in \PS} \fn(\pa)$.
For a given number $r \geq 0$, set $\PS_r = \LvSetY{\fn}{r} \cap \PS$.
We develop a randomized algorithm to compute $\alpha$.

Set $\PS_0 = \PS$. In the $i$\th iteration, the algorithm chooses a
random point $\pa_i \in \PS_{i-1}$ and computes $r_i = \fn(\pa_i)$.
Next, we determine $\PS_{r_i}$ using \lemref{decision}. In doing so,
we modify the separation oracle of \lemref{decision} to store the
collection of queries $S_i \subseteq \PS$ that satisfy $\fn(\ps) = r_i$
for all $\ps \in S_i$. We set $\PS_{i+1} = \PS_{r_i} \setminus S_i$.
Observe that all points $\pa \in \PS_{i+1}$ have $\fn(\pa) < r_i$.
The algorithm continues in this fashion until we reach an iteration
$j$ in which $\cardin{\PS_{j+1}} \leq 1$. If $\PS_{j+1} = \{\pb\}$ for
some $\pb \in \PS$, output $\pb$ as the desired point minimizing $\fn$.
Otherwise $\PS_{j+1} = \varnothing$, implying that
$\PS_{r_j} = S_j$, and the algorithm outputs any point in the set
$S_j$.

\myparagraph{Analysis.}
We analyze the running time of the algorithm. To do so, we argue
that the algorithm invokes the algorithm in \lemref{decision} only
a logarithmic number of times.

\begin{lemma}
  \lemlab{jumps}
  In expectation, the above algorithm terminates after $O(\log n)$
  iterations.
\end{lemma}
\begin{proof}
    Let $V = \Set{\fn(\pa)}{\pa \in \PS}$ and $N = \cardin{V}$. For a
    number $r$, define $V_r = \Set{i \in V}{i \leq r}$. Notice that we
    can reinterpret the algorithm described above as the following
    random process. Initially set $r_0 = \max_{i \in V} i$. In the
    $i$\th iteration, choose a random number $r_i \in
    V_{r_{i-1}}$. This process continues until we reach an iteration
    $j$ in which $\cardin{V_{r_j}} \leq 1$.

  We can assume without loss of generality that
  $V = \{1, 2, \ldots, N\}$. For an integer $i \leq N$,
  let $T(i)$ be the expected number of iterations needed for the
  random process to terminate on the set $\{1, \ldots, i\}$. We have
  that $T(i) = 1 + \frac{1}{i-1} \sum_{j=1}^{i-1} T(i-j)$, with
  $T(1) = 0$. This recurrence solves to $T(i) = O(\log i)$.
  As such, the algorithm repeats this random process
  $O(\log N) = O(\log n)$ times in expectation.
\end{proof}

\begin{lemma}
    \lemlab{discrete:min} Let $\PS$ be a set of $n$ points in $\Re^2$
    and let $\fn : \Re^2 \to \Re$ be a convex function. The point in
    $\PS$ minimizing $\fn$ can be computed using
    $O(\indexX{\PS} \log^2 n)$ evaluations to $\fn$ or its
    derivative. The bound on the number of evaluations holds in
    expectation. If $\TX$ is the time needed to evaluate $\fn$ or its
    derivative, the algorithm can be implemented in
    $O(n\log^3 n\log\log n + \TX\cdot \indexX{\PS} \log^2 n)$ expected
    time.
\end{lemma}
\begin{proof}
  The result follows by combining \lemref{decision} and
  \lemref{jumps}.
\end{proof}

\subsection{The discrete geometric median}

Let $\PS$ be a set of $n$ points in $\Re^d$. For all $x \in \Re^d$, define
the function $\fn(x) = \sum_{\pb \in \PS - x} \normX{x - \pb}$.
The \emphi{discrete geometric median} is defined as the point in $\PS$
minimizing the quantity $\min_{\pa \in \PS} \fn(\pa)$.

Note that $\fn$ is convex, as it is the sum of convex functions.
Furthermore, given a point $\pnt$, we can compute $\fn(\pnt)$ and the
derivative of $\fn$ at $\pnt$ in $O(n)$ time. As such, by
\lemref{discrete:min}, we obtain the following.

\begin{lemma}
  \lemlab{discrete:med}
  Let $\PS$ be a set of points in $\Re^2$. Then
  the discrete geometric median of $\PS$ can be computed in
  $O(n\log^2 n \cdot (\log n \log\log n + \indexX{\PS}))$
  expected time.
\end{lemma}

\begin{remark}
  For a set of $n$ points $\PS$ chosen uniformly at random from the unit
  square, it is known that in expectation $\indexX{\PS} = \Theta(n^{1/3})$
  \cite{ab-lcc-09}.
  As such, the discrete geometric median for such a random set $\PS$
  can be computed in $O(n^{4/3} \log^2 n)$ expected time.
\end{remark}

\section{Conclusion and open problems}

In this paper we presented various algorithms for classifying
points with oracle access to an unknown convex body. As far as the
authors are aware, this exact problem has not been studied within the
computational geometry community previously. However, since the problem
is closely related to active learning and has further applications (such as
discretely minimizing a convex function), we believe that this is an
interesting problem to study. We now pose some open problems.

\begin{compactenumA}
  \item Develop a more natural instance-optimal algorithm in 2D that
  improves upon the $O(\log^2 n)$ approximation. Alternatively, develop
  algorithms in which the number of queries is parameterized by different
  functions of the input instance.
  \smallskip%

  \item An algorithm in 3D that is instance-optimal up to some
  additional factors (see the beginning of \apndref{ex:sep:pr} for
  the definition of the separation price in higher dimensions).
  \smallskip%

  \item Any results beyond three dimensions is unknown. The greedy
  algorithm (\thmref{greedy-method} and \thmref{greedy-method-3d})
  easily extends to $\Re^d$. However the analysis in higher dimensions
  will most likely reveal that the algorithm makes (ignoring logarithmic
  factors) of the order of ${\indexX{\PS}}^{O(d)}$ queries, which is only
  interesting when $\indexX{\PS}$ is much smaller than $n$.
  \smallskip%

  \item In \lemref{discrete:med} we gave a randomized algorithm for computing
  the discrete geometric median in expected time $\tldO(n \cdot \indexX{\PS})$
  (where $\tldO$ hides logarithmic factors in $n$). The bottleneck of the
  algorithm was in the computation of the function and the gradient, which
  naively requires $O(n)$ time. Is it possible to speed up the gradient
  computation by introducing additional randomization or optimization
  techniques? Improving the running time further is an open problem.
\end{compactenumA}

\BibTexMode{%
   \SubmitVer{%
      \bibliographystyle{spmpsci}%
   }%
   \RegVer{%
      \bibliographystyle{alpha}%
   }%
   \bibliography{resistance}%
}

\BibLatexMode{\printbibliography}

\appendix

\section{Expected separation price for random points}
\apndlab{ex:sep:pr}

We first extend the notion of separation price (see
\secref{lower:bound}) to higher dimensions. For a closed
convex $d$-dimensional polytope $F$, we let $f_k(F)$ denote
the number of $k$-dimensional faces of $F$.

\begin{definition}[Separation price in higher dimensions]
  \deflab{lower:bound:high}
  Let $\PS$ be a set of points and $\body$ be a convex body in
  $\Re^d$. The inner fence $\Fin$ is a closed convex $d$-dimensional
  polytope with the minimum number of vertices, such that
  $\Fin \subseteq \body$ and $\body \cap \PS = \Fin \cap \PS$.
  Similarly, the outer fence $\Fout$ is a closed convex
  $d$-dimensional polytope with the minimum number of facets, such
  that $\body \subseteq \Fout$ and $\body \cap \PS = \Fout \cap \PS$.
  The separation price is defined as
  $\priceY{\PS}{\body} = f_0(\Fin) + f_{d-1}(\Fout)$.
\end{definition}

By extending the argument of \lemref{lower:bound} to use
\defref{lower:bound:high}, one can prove the following.

\begin{lemma}
    \lemlab{lower:bound:high} %
    Given a point set $\PS$ and a convex body $\body$ in $\Re^d$, any
    algorithm that classifies the points of $\PS$ in relation to
    $\body$, must perform at least $\priceY{\PS}{\body}$ separation
    oracle queries.
\end{lemma}

Informally, for any fixed convex body $\body$ and a set of $n$ points
$\PS$ chosen uniformly at random from the unit cube, the separation
price is sublinear (approaching linear as the dimension increases).

\begin{lemma}
    \lemlab{ex:sep:pr}%
    Let $\PS$ be a set of $n$ points chosen uniformly at random from
    the unit cube $[0,1]^d$, and let $\body$ be a convex body in
    $\Re^d$, with $\VolX{\body} \geq c$ for some constant $c \leq 1$.
    Then $\Ex{\priceY{\PS}{\body}} = O(n^{1 - 2/(d+1)})$, where $O$
    hides constants that depend on $d$ and $\body$.
\end{lemma}
\begin{proof}
    It is known that for convex bodies $\body$, the expected number of
    vertices of the convex hull of $\PS \cap \body$ is
    $O(n^{1 - 2/(d+1)})$.  Indeed, since $\VolX{\body} \geq c$, the
    expected number of points of $\PS$ that fall inside $\body$ is
    $m = \Theta(n)$ (and these bounds hold with high probability by
    applying any Chernoff-like bound). It is known that for $m$ points
    chosen uniformly at random from $\body$, the expected size of the
    convex hull of points inside $\body$ is
    $O(m^{1 - 2/(d+1)}) = O(n^{1 - 2/(d+1)})$ \cite{b-rpcba-07}.  This
    readily implies that $\Ex{f_0(\Fin)} = O(n^{1 - 2/(d+1)})$.

    To bound $\Ex{f_{d-1}(\Fout)}$, we apply a result of Dudley
    \cite{d-mescs-74} which states the following. Given a convex body
    $\body$ and a parameter $\eps > 0$, there exists a
    convex body $D$, which is a polytope
    formed by the intersection of $O(\eps^{-(d-1)/2})$ halfspaces,
    such that $\body \subseteq D \subseteq (1+\eps)\body$, where
    $(1+\eps)\body = \Set{\pnt \in \Re^d}{\exists \pntq \in \body : \|
       \pnt - \pntq \| \leq \eps}$.

    We claim that the number of points of $\PS$ that fall inside $D
    \setminus \body$, plus the number of halfspaces defining $D$, is
    an upper bound on the size of the outer fence. Indeed, for each
    point $\pnt$ that falls in inside $D \setminus \body$, let
    $\pntq$ be its nearest neighbor in $\body$ (naturally $\pntq$ lies
    on $\BX{\body}$). Let $\LineA_\pnt$ be the hyperplane that is
    perpendicular to the segment $\pnt\pntq$ and passing through the
    midpoint of $pq$. Next, let $\LineA^+_\pnt$ be the halfspace
    bounded by $\LineA_\pnt$ such that $\body \subseteq \LineA^+_\pnt$.
    If $H$ is the collection of $O(\eps^{-(d-1)/2})$ halfspaces defining
    $D$, then it is easy to see that the polytope defined by
    \begin{equation*}
        \Bigl({\bigcap_{\pnt \in \PS \cap (D \setminus \body)} \LineA^+_\pnt}\Bigr)
        \ts%
        \bigcap
        \,
        \Bigl({\bigcap_{\LineA^+ \in H} \LineA^+}\Bigr)
    \end{equation*}
    separates the boundary of $\body$ from $\PS \setminus C$ (i.e., it
    is an outer fence). See \figref{demonstration}.

    \begin{figure}[h]
        \includegraphics[page=1,width=0.3\linewidth]%
        {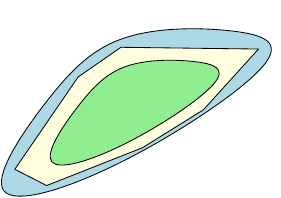}
        \hfill%
        \includegraphics%
        [page=3,width=0.3\linewidth]%
        {figs/separation}
        \hfill%
        \includegraphics%
        [page=4,width=0.3\linewidth]%
        {figs/separation}
        \caption{}
        \figlab{demonstration}
    \end{figure}

    We now bound the size of this inner fence. Since
    $\VolX{D} - \VolX{\body} \leq \VolX{(1+\eps)\body} - \VolX{\body}
    \leq O(\eps)$, we have that
    $\Ex{\cardin{\PS \cap (D \setminus \body)}} = O(\eps n)$.
    Combining both inequalities,
    \begin{align*}
      \Ex{f_{d-1}(\Fout)}
      \leq \Ex{\cardin{\PS \cap (D \setminus \body)}} + O(\eps^{-(d-1)/2})
      = O\pth{\eps n + \frac{1}{\eps^{(d-1)/2}}}.
    \end{align*}
    Choose $\eps = 1/n^{2/(d+1)}$ to balance both terms, so that
    $\Ex{f_{d-1}(\Fout)} = O(n^{1 - 2/(d+1)})$.
\end{proof}

The next Lemma shows that the bound of \lemref{ex:sep:pr} is tight
in the worst case.

\begin{lemma}
    Let $\PS$ be a set of $n$ points chosen uniformly at random from
    the hypercube $[-2,2]^d$, and let $\body$ be a unit radius ball
    centered at the origin. Then
    $\Ex{\priceY{\PS}{\body}} \geq \Ex{f_0(\Fin)} = \Omega(n^{1 - 2/(d+1)})$,
    where $\Omega$ hides constants depending on $d$.
\end{lemma}

\begin{proof}
    For a parameter $\delta$ to be chosen, let $Q \subseteq \BX{\body}$
    be a maximal set of points such that:
    \begin{compactenumi}
      \item for any $\pnt \in \BX{\body}$, there is a point
      $\pntq \in Q$ such that $\| \pnt - \pntq \| \leq \delta$, and
      \item for any two points
      $\pnt, \pntq \in Q$, $\| \pnt - \pntq \| \geq \delta$.
    \end{compactenumi}
    Note that $\cardin{Q} = \Omega(1/\delta^{d-1})$. For each $\pnt
    \in Q$, we let $\gamma_\pnt$ be the spherical cap that is
    ``centered'' at $\pnt$ (in the sense that the center of the base of
    $\gamma_\pnt$, $\pnt$, and the origin are collinear) and has base
    radius $2\delta$.  Let $\Gamma = \Set{\gamma_\pnt}{\pnt \in Q}$.
    By construction, the caps of $\Gamma$ cover the surface of
    $\body$.

    By setting $\delta = 1/n^{1/(d+1)}$, we claim that for
    each cap $\gamma \in \Gamma$, in expectation $\Omega(1)$ points of
    $\PS$ fall inside $\gamma$. This implies that there must be a
    vertex of the inner fence inside $\gamma$, and this holds for
    all caps in $\Gamma$. As such, the size of the inner fence is at
    least
    $\cardin{Q} = \Omega(1/\delta^{d-1}) = \Omega(n^{1 - 2/(d+1)})$.

    \begin{figure}[h]
        \centerline{%
           \includegraphics{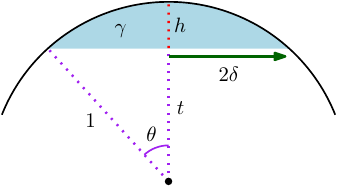}%
        }
        \caption{}%
        \figlab{what:what}
    \end{figure}

    To prove the claim, for all $\gamma \in \Gamma$, we show that
    $\VolX{\gamma} = \Omega(1/n)$, and hence
    $\Ex{\cardin{\PS \cap \gamma}} = \Omega(1)$. By construction, the
    cap has a polar angle of $\theta = \Omega(\delta)$, see
    \figref{what:what}. Indeed, we have that
    $\theta \geq \sin(\theta) = 2\delta$ for $\theta \in [0,\pi/2]$
    (which holds when $n$ is sufficiently large). Let $t$ denote the
    distance from the origin to the center of the base of $\gamma$.
    Then the height $h$ of the spherical cap is
    $h = 1 - t = 1 - \cos(\theta) \geq \theta^2/6 = \Omega(\delta^2)$
    (using the inequality $\cos(x) \leq 1 - x^2/6$).  Since the volume
    of the base of $\gamma$ is $\Omega(\delta^{d-1})$, we have that
    $\VolX{\gamma} = \Omega(h \delta^{d-1}) = \Omega(\delta^{d+1}) =
    \Omega(1/n)$, as required.
\end{proof}

\end{document}